\documentclass{lmcs}
\pdfoutput=1

\usepackage{lastpage}
\lmcsdoi{16}{2}{1}
\lmcsheading{}{\pageref{LastPage}}{}{}%
{Sep.~17,~2018}{Apr.~01,~2020}{}

\usepackage{hyperref}

\usepackage{graphicx}
\usepackage{tikz}
\usepackage{bm}
\usepackage{ifpdf}
\usepackage{amsmath}
\usepackage{amssymb}

\usepackage{pseudocode}
\usepackage{chngcntr}

\counterwithout{pseudocode}{section}
\newcounter{newpseudonum}[pseudocode]

\ifpdf
  \providecommand{\refline}[1]{\hyperref[#1]{(\ref*{#1})}}
\else
  \providecommand{\refline}[1]{\ref*{#1}}
\fi

\renewcommand{\RETURN}[1]{\ifthenelse{\equal{#1}{} }{\mbox{\bfseries return}}{\mbox{\bfseries return}#1}}
\newcommand{\FUNCTION}[2]{\mbox{\bfseries proc }\mbox{\textsc{#1}}\left(\ensuremath{#2}\right)\\}

\newlength{\pcodewidth}
\newenvironment{code}[1]{
\begin{Sbox}
\!\!\begin{minipage}{#1}%\pcodewidth}
\bfseries
\noindent
\scriptsize
$$
\begin{array}{@{\hspace*{1ex}}lr@{}}
}{
\end{array}
$$
\end{minipage}\vspace{-2mm}
\end{Sbox}
\shadowbox{\TheSbox}{}
}
\usetikzlibrary{decorations,arrows,shapes}
\usetikzlibrary{positioning}
\usetikzlibrary{arrows,automata}
\usepackage{url}
\DeclareMathAlphabet{\mathpzc}{OT1}{pzc}{m}{it}

\newcommand{\ii}[2]{{\mathpzc{#1}\mathpzc{#2}}_{ii}}
\newcommand{\ip}[1]{{\mathpzc{#1}}_{ip}}
\renewcommand{\pi}[2]{{\mathpzc{#1}\mathpzc{#2}}_{pi}}
\newcommand{\pp}[1]{{\mathpzc{#1}}_{pp}}

\newcommand{\Lin}{\mathsf{Lin}}
\newcommand{\Den}{\mathsf{Den}}
\newcommand{\Dis}{\mathsf{Dis}}
\newcommand{\Unb}{\mathsf{Unb}}

\newcommand{\allr}{\mathfrak{R}}
\newcommand{\alli}{\mathfrak{I}}
\newcommand{\allm}{\mathfrak{M}}
\newcommand{\allp}{\mathfrak{P}}

\newtheorem{theorem}[thm]{Theorem}
\newtheorem{lemma}[thm]{Lemma}

\newtheorem{definition}[thm]{Definition}

\begin{document}

\title[A Theory of Points and Intervals (II)]{An Integrated First-Order Theory of Points and Intervals over Linear Orders (Part II)}

\author[W.~Conradie]{Willem Conradie\rsuper{a}}	%required
\address{\lsuper{a}School of Mathematics, University of Witwatersrand\\
        Johannesburg, South Africa}	%required
\email{willem.conradie@wits.ac.za}  %optional
\thanks{The research of the first author was partially supported by grant number 81309 of the South African National Research Foundation.}	%optional

\author[S.~Durhan]{Salih Durhan\rsuper{b}}	%optional
\address{\lsuper{b}D4C Research and Development\\
	Istanbul, Turkey}	%optional
\email{salihdurhan@d4c.ai }  %optional
%\thanks{thanks 2, optional.}	%optional

\author[G.~Sciavicco]{Guido Sciavicco\rsuper{c}}	%optional
\address{\lsuper{c}Department of Mathematics and Computer Science\\
        University of Ferrara\\
        Ferrara, Italy}	%optional
\email{guido.sciavicco@unife.it}  %optional
\thanks{We would like to thank Dr. Davide Bresolin, of the University of Padova, for his help. This paper is the second part of a 2-parts paper.}	%optional

%% etc.

%% required for running head on odd and even pages, use suitable
%% abbreviations in case of long titles and many authors:

%% mandatory lists of keywords and classifications:
\subjclass{[Theory of computation]: Logic --- Modal and temporal logics; [Computing methodologies]: Artificial intelligence --- Knowledge representation and reasoning ---  Temporal reasoning}
%\titlecomment{This paper is the first part of a two-parts paper.}
%%%%%%%%%%%%%%%%%%%%%%%%%%%%%%%%%%%%%%%%%%%%%%%%%%%%%%%%%%%%%%%%%%%%%%%%%%%

%% the abstract has to PRECEED the command \maketitle:
%% be sure not to issue the \maketitle command twice!

\begin{abstract}
  There are two natural and well-studied approaches to temporal ontology and reasoning: point-based and interval-based. Usually, interval-based temporal reasoning deals with points as a particular case of duration-less intervals. A recent result by Balbiani, Goranko, and Sciavicco presented an explicit two-sorted point-interval temporal framework in which time instants (points) and time periods (intervals) are considered on a par, allowing the perspective to shift between these within the formal discourse. We consider here two-sorted first-order languages based on the same principle, and therefore including relations, as first studied by Reich, among others, between points, between intervals, and inter-sort. We give complete classifications of its sub-languages in terms of relative expressive power, thus determining how many, and which, are the intrinsically different extensions of two-sorted first-order logic with one or more such relations. This approach roots out the classical problem of whether or not points should be included in a interval-based semantics. In this Part II, we deal with the cases of all dense and the case of all unbounded linearly ordered sets.
\end{abstract}

\maketitle

\section{Introduction}\label{sec:intro}

The relevance of temporal logics in many theoretical and applied areas of computer science and AI, such as theories of action and change, natural language analysis and processing, and constraint satisfaction problems, is widely recognized. While the predominant approach in the study of temporal reasoning and logics has been based on the assumption that time points (instants) are the primary temporal ontological entities, there has also been significant activity in the study of interval-based temporal reasoning and logics over the past two decades. The variety of binary relations between intervals in linear orders was first studied systematically by Allen~\cite{Allen87,Allen:1983:MKA,JLOGC::AllenF1994}, who explored their use in systems for time management and planning. Allen's work and much that follows from it is based on the assumption that time can be represented as a dense line, and that points are excluded from the semantics. At the modal level, Halpern and Shoham~\cite{JACM::HalpernS1991} introduced the multi-modal logic HS that comprises modal operators for all possible relations (known as Allen's relations~\cite{Allen:1983:MKA}) between two intervals in a linear order, and it has been followed by a series of publications studying the expressiveness and decidability/undecidability and complexity of the fragments of HS, e.g.,~\cite{tcs2014,amai2014}. Many studies on interval logics have considered the so-called `non-strict' interval semantics, allowing point-intervals (with coinciding endpoints) along with proper ones, and thus encompassing the instant-based approach, too; more recent ones, instead, started to treat pure intervals only. Yet, little has been done so far on the formal treatment of both temporal primitives, points and intervals, in a unified two-sorted framework. A detailed philosophical study of both approaches, point-based and interval-based, can be found in~\cite{vB91} (see also~\cite{590368}). A similar mixed approach has been studied in~\cite{DBLP:journals/ci/AllenH89}. ~\cite{DBLP:journals/cj/MaH06} contains a study of the two sorts and the relations between them in dense linear orders. More recently, a modal logic that includes different operators for points and interval has been presented in~\cite{DBLP:journals/entcs/BalbianiGS11}.

\medskip

The present paper provides a systematic treatment of point and interval relations (including equality between points and between intervals treated on the same footing as the other relations) at the first-order level. Our work is motivated, among other observations, by the fact that natural languages incorporate both ontologies on a par, without assuming the primacy of one over the other, and have the capacity to shift the perspective smoothly from instants to intervals and vice versa within the same discourse, e.g.: {\em when the alarm goes on, it stays on until the code is entered}, which contains two instantaneous events and a non-instantaneous one. Moreover, there are various temporal scenarios which neither of the two ontologies alone can grasp properly since neither the treatment of intervals as the sets of their internal points, nor the treatment of points as `instantaneous' intervals, is really adequate. The technical identification of intervals with sets of their internal points, or of points as instantaneous intervals leads also to conceptual problems like the confusion of events and fluents. Instantaneous events are represented by time intervals and should be distinguished from instantaneous holding of fluents, which are evaluated at time points: therefore, the point $a$ should be distinguished from the interval $[a,a]$, and the truths in these should not necessarily imply each other. Finally, we note that, while differences in expressiveness have been found between the strict and non-strict semantics for some interval logics (see~\cite{ijcai11}, for example), so far, no distinction in the decidability of the satisfiability has been found. Therefore, we believe that an attempt to systemize the role of points, intervals, and their interaction, would make good sense not only from a purely ontological point of view, but also from algorithmic and computational perspectives.

\medskip

\noindent{\bf Previous Work and Motivations.} As presented in the early work of van Benthem~\cite{vB91} and Allen and Hayes~\cite{Allen85}, interval temporal reasoning can be formalized as an extension of first-order logic with equality with one or more relations, and the properties of the resulting language can be studied; obviously, the same applies when relations between points are considered too. In this paper we ask the question: interpreted over linear orders, how many and which expressively different languages can be obtained by enriching first-order logic with relations between intervals,  between points, and between intervals and points? Since, as we shall see, there are 26 different relations (including equality of both sorts) between points, intervals, and points and intervals, $2^{26}$ is an upper bound on this number. (It is worth noticing that in~\cite{DBLP:journals/cj/MaH06} the authors distinguish 30 relations, instead of 26; this is due to the fact that the concepts of the point $a$ {\em starting} the interval $[a,b]$ and {\em meeting} it are considered to be different.) However, since certain relations are definable in terms of other ones, the actual number is less and in fact, as we shall show, much less. The answer also depends on our choices of certain semantic parameters, specifically, the class of linear orders over which we construct our interval structures. In this paper, in Part I~\cite{DBLP:journals/lmcs/ConradieDS18}, we consider the classification problem relative to:

\begin{enumerate}[label=\emph{(\roman*)}]
\item the class of all linear orders;
\item the class of all  weakly discrete linear orders (i.e., orders in which every point with a successor/predecessor has an immediate one).
\end{enumerate}

\noindent In Part II of this paper we consider:
\begin{enumerate}[resume,label=\emph{(\roman*)}]
\item the class of all dense linear orders;
\item the class of all unbounded linear orders;
%\item[\it (v)] the class of all finite linear orders.
\end{enumerate}

\medskip

Apart from the intrinsic interest and naturalness of this classification problem, its outcome has some important repercussions, principally in the reduction of the number of cases that need to be considered in other problems relating to these languages. For example, it reduces the number of representation theorems that are needed: given the {\em dual} nature of time intervals (i.e., they can be abstract first-order individuals with specific characteristics, or they can be defined as ordered pairs over a linear order), one of the most important problems that arises is the existence or not of a {\em representation theorem}. Consider any class of linear orders: given a specific extension of first-order logic with a set of interval relations (such as, for example, {\em meets} and {\em during}), does there exist a set of axioms in this language which would constrain (abstract) models of this signature to be isomorphic to concrete ones?  Various representation theorems exist in the literature for languages that include interval relations only: van Benthem~\cite{vB91}, over rationals and with the interval relations {\em during} and {\em before}, Allen and Hayes~\cite{Allen85}, for the dense unbounded case without point intervals and for the relation {\em meets}, Ladkin~\cite{Ladkin}, for point-based structures with a quaternary relation that encodes meeting of two intervals, Venema~\cite{JLOGC::Venema1991}, for structures with the relations {\em starts} and {\em finishes}, Goranko, Montanari, and Sciavicco~\cite{GMS03}, for linear structures with {\em meets} and {\em met-by}, Bochman~\cite{DBLP:journals/ndjfl/Bochman90}, for point-interval structures, and Coetzee~\cite{Coetzee:MThesis} for dense structure with {\em overlaps} and {\em meets}. Clearly, if two sets of relations give rise to expressively equivalent languages, two separate representations theorems for them are not needed.  In which cases are representation theorems still outstanding? Preliminary works that provide similar classifications appeared in~\cite{DBLP:conf/caepia/ConradieS11} for first-order languages with equality and only interval-interval relations, and in~\cite{time2012} for points and intervals (with equality between intervals treated on a par with the other relations) but only over the class of all linear orders. Finally, a complete study of first-order interval temporal logics enables a deeper understanding of their modal counterparts based on their shared relational semantics.

\medskip

\noindent{\bf Structure of the paper.} This paper is structured as follows. Section~\ref{sec:basics} provides the necessary preliminaries, along with an overview of the general methodology used in this paper. Part I of this paper dealt with definability and undefinability in the classes $\Lin$ and $\Dis$, from which we start in order to tackle, in Section~\ref{sec:den}, the study the expressive power of the language by analyzing the definability properties of each basic relation in the class $\Den$, and in Section~\ref{sec:den-incomp} the corresponding undefinability results in this case. Then, in Section~\ref{sec:unb} and Section~\ref{sec:unb-incomp}, respectively, we present the same analysis in the unbounded case, before concluding. It is worth reminding that most of the results presented here are a consequence of those presented in Part I, to which we shall refer whenever necessary.

\section{Basics}\label{sec:basics}

\subsection{Syntax and semantics}

\begin{table}[t]
\centering
\begin{tikzpicture}[scale=1]
% \tikzstyle{every node}=[font=\tiny]
% \tikzstyle{every node}=[font=\scriptsize]
\tikzstyle{every node}=[font=\small]
% \draw (0,0)node(op){\hs \bf operators};
\draw (0,0)node(op){};
% \draw (4,0)node{\bf Corresponding relations};
%
% \draw (op) ++(-1.5,-.25) -- ++(12,0);

\draw (op) ++(0,-.75)node(meets){};
\draw (meets)++(0,-.75)node(later){};
\draw (meets)++(0,-1.5)node(starts){};
\draw (meets)++(0,-2.25)node(finishes){};
\draw (meets)++(0,-3)node(during){};
\draw (meets)++(0,-3.75)node(overlaps){};
\draw (meets)++(-1.0,0)node[right](Ra){$(\mbox{{\em meets,m}})\,\,[a,b]~\ii{3}{4}~[c,d] \Leftrightarrow b=c$};
\draw (meets)++(-1.0,-.75)node[right](Rl){$(\mbox{{\em before,b}})\,\,[a,b]~\ii{4}{4}~[c,d] \Leftrightarrow b < c$};
\draw (meets)++(-1.0,-1.5)node[right](Rs){$(\mbox{{\em starts,s}})\,\,[c,d]~\ii{1}{4}~[a,b] \Leftrightarrow a=c, d < b$};
\draw (meets)++(-1.0,-2.25)node[right](Rf){$(\mbox{{\em finishes,f}})\,\,[c,d]~\ii{0}{3}~[a,b] \Leftrightarrow b=d, a < c$};
\draw (meets)++(-1.0,-3)node[right](Rd){$(\mbox{{\em during,d}})\,\,[c,d]~\ii{0}{4}~[a,b] \Leftrightarrow a < c, d < b$};
\draw (meets)++(-1.0,-3.75)node[right](Ro){$(\mbox{{\em overlaps,o}})\,\,[a,b]~\ii{2}{4}~[c,d] \Leftrightarrow a < c < b < d$};
\draw[red,|-|] (meets) ++(7.6,.75)node[above](a){$a$} -- ++(2,0)node[above](b){$b$};
\draw[dashed,red,help lines,thick] (a) -- ++(0,-5.25);
\draw[dashed,red,help lines,thick] (b) -- ++(0,-5.25);
\draw[|-|] (b) ++(0,-1) ++(0,0)node[above](Ac){$c$}
-- ++(1,0)node[above](Ad){$d$};
\draw[|-|] (b) ++(0,-1) ++(.5,-.75)node[above](Lc){$c$}
-- ++(1,0)node[above](Ld){$d$};
\draw[|-|] (a) ++(0,-1) ++(0,-1.5)node[above](Bc){$c$}
-- ++(.5,0)node[above](Bd){$d$};
\draw[|-|] (b) ++(0,-1) ++(-.5,-2.25)node[above](Ec){$c$}
-- ++(.5,0)node[above](Ed){$d$};
\draw[|-|] (a) ++(0,-1) ++(.5,-3)node[above](Dc){$c$}
-- ++(1,0)node[above](Dd){$d$};
\draw[|-|] (a) ++(0,-1) ++(1,-3.75)node[above](Oc){$c$}
-- ++(2,0)node[above](Od){$d$};
% $[a,b] R_A [c,d]$ & iff $b=c$ \tikz\draw[|-|] (0,0) -- ++(1,0); \\
\end{tikzpicture}
\caption{Interval-interval relations, a.k.a.\ Allen's relations. The equality relation is not depicted.}
\label{ii:relations}
%\vspace{-0.5cm}
\end{table}

Given a linear order $\mathbb D=\langle D,<\rangle$, we call the elements of $D$ \emph{points} (denoted by $a,b,\ldots$) and define an {\em interval} as an ordered pair $[a,b]$ of points in $D$, where $a<b$. Abstract intervals will be denoted by $I,J,\ldots,$ and so on.
Now, as we have mentioned above, there are 13 possible relations, including equality, between any two intervals. From now on, we
call these {\em interval-interval} relations. Besides equality, there are 2 different relations that may hold between any two points ({\em before}
and {\em after}), called hereafter {\em point-point} relations, and 5 different relations that may hold between a point and an interval and vice-versa: we call those {\em interval-point} and {\em point-interval} relations, respectively, and we use the term {\em mixed} relations to refer to them indistinctly. Interval-interval relations are exactly Allen's relations~\cite{Allen:1983:MKA}; point-point relations are the classical relations on a linear order, and mixed relations will be explained below. Traditionally, interval relations are represented by the initial letter of the description of the relation, like $m$ for {\em meets}. However, when one considers more relations (like point-point and point-interval relations) this notation becomes confusing, and even more so in the presence of more relations, e.g.\ when one wants to consider interval relations over a {\em partial order}\footnote{This paper is focused on linear orders only; nevertheless, it is our intention to complete this study to include the treatment of partial orders also, and, at this stage, we want to make sure that we will be able to keep the notation consistent.}. We introduce the following notation to resolve this issue: an interval $[a,b]$ induces a
partition of $\mathbb D$ into five regions (see~\cite{lig91}): region 0 which contains all
points less than $a$, region 1 which contains $a$ only,
region 2 which contains all the points strictly between $a$ and $b$, region 3 which contains
only $b$ and region 4 which contains the points greater than $b$. Likewise, a point $c$ induces a partition of $\mathbb D$ into 3 pieces: region 0 contains all the points less than $c$, region 2 contains only $c$, and region 4 contains all the points greater than $c$. Interval-interval relations will be denoted by $I\ii{k}{k'}J$ (where the subscript $_{ii}$ refers to interval-interval relations), where $k,k'\in\{0,1,2,3,4\}$, and $k$ represent the region of the partition induced by $I$ in which the left endpoint of $J$ falls, while $k'$ is the region of the same partition in which the right endpoint of $J$ falls; for example, $I\ii{3}{4}J$ is exactly Allen's relation {\em meets}. Similarly, interval-point relations will be denoted by $I\ip{k\,}a$ (where the subscript $_{ip}$ stands for interval-point relations), where $k$ represents the position of $a$ with respect to $I$; for example, $I\ip{4}a$ is the relation {\em before}. Analogously, point-point relations will be denoted by the symbol $\pp{k\,}$, and point-interval relations by the symbol $\pi{k}{k'}$. For point-point relations it is more convenient to use $<$ instead of $\pp{4}$, and $>$ instead of $\pp{0}$.  In Tab.~\ref{ii:relations} we show six of the interval-interval relations, along with its original nomenclature and symbology, and in Tab.~\ref{ip:relations} we show the interval-point relations. Finally, we consider one equality per sort, using  $=_i$ to denote $\ii 13$ (equality between intervals), and $=_p$ to denote $\pp 2$ (the equality between points). Now, given any of the mentioned relations $r$, its inverse, generically denoted by $\bar r$, can be obtained by inverting the roles of the objects in the case of non-mixed relations; for example, the inverse of the relation $\ii{2}{2}$ (Allen's relation {\em contains}) is the relation $\ii{0}{4}$ (Allen's relation {\em during}). On the other hand, mixed relations present a different situation: the inverse of a point-interval relation is an interval-point relation; thus, for example, the inverse of $\ip{3}$ is $\pi{0}{2}$. Finally, notice that some combinations are forbidden: for instance, the relation $\pi{2}{2}$ makes no sense, as all intervals have a non-zero extension.

\begin{definition}
We shall denote by: $\allr$ the set of all above described relations; $\alli\subset\allr$ the subset of all 13 interval-interval relations (Allen's relations) including the relation $=_i$; $\allm\subset\allr$ the subset of all mixed relations; $\allp\subset\allr$ the subset of all point-point relations including the relation $=_p$. Clearly, $\allr=\alli \bigcup \allm \bigcup \allp$.
\end{definition}

\begin{table}[t]
\centering
\begin{picture}(100,100)
\color{red}
\put(60,90){\line(1,0){40}}
\put(60,87){\line(0,1){6}}
\put(100,87){\line(0,1){6}}
\put(58,94){\small{$a$}}
\put(98,94){\small{$b$}}
\color{black}
\put(-70,74){\small{$[a,b]~\ip{3}~c \Leftrightarrow b=c$}}
\put(98,77){$\cdot$}
\put(97,81){\small{$c$}}
\put(-70,58){\small{$[a,b]~\ip{4}~c \Leftrightarrow b<c$}}
\put(120,60){$\cdot$}
\put(119,64){\small{$c$}}
\put(-70,42){\small{$[a,b]~\ip{2}~c \Leftrightarrow a<c<b$}}
\put(80,45){$\cdot$}
\put(79,49){\small{$c$}}
\put(-70,28){\small{$[a,b]~\ip{1}~c \Leftrightarrow a=c$}}
\put(59,29){$\cdot$}
\put(58,35){\small{$c$}}
\put(-70,12){\small{$[a,b]~\ip{0}~c \Leftrightarrow c<a$}}
\put(40,10){$\cdot$}
\put(39,15){\small{$c$}}
\multiput(60,15)(0,8){9}{\line(0,1){3}}
\multiput(100,15)(0,8){9}{\line(0,1){3}}
\end{picture}
%\vspace{-0.5cm}
\caption{Interval-point relations.}
\label{ip:relations}
\end{table}

\begin{definition}
In the following, we denote by:

\begin{enumerate}[label=\emph{(\roman*)}]
\item $\Lin$ the class of all linear orders;
\item $\Den$ the class of all dense linear orders, that is, the class of all linear orders where there exists a point in between any two distinct points;
\item $\Dis$ the class of all weakly discrete linear orders, that is, the class of all linear orders where each point, other than the least (resp., greatest) point, if there is one, has a direct predecessor (resp., successor) -- by a {\em direct predecessor} of $a$ we of course mean a point $b$ such that $b < a$ and for all points $c$, if $c < a$ then $c \leq b$, and the notion of a {\em direct successor} is defined dually;
\item $\Unb$ the class of all unbounded linear orders, that is, the class of all linear order such that for every point $a$ there exists a point $b>a$ and a point $c<a$.
%\item $\Fin$ the class of all finite linear orders.
\end{enumerate}
\end{definition}

\begin{definition}
Given a linear order $\mathbb D$, and given the set $\mathbb I(\mathbb D)=\{[a,b]\mid a,b\in \mathbb D, a<b\}$ of all intervals built on $\mathbb D$:
\begin{itemize}
\item a \emph{concrete interval structure of signature $S$} is a relational structure $\mathcal{F} = \langle \mathbb I(\mathbb D), r_1,$ $r_2, \ldots,r_n \rangle$, where $S = \{r_1, \ldots,r_n \} \subseteq \alli$, and
\item a \emph{concrete point-interval structure of signature $S$} is a two-sorted relational structure $\mathcal{F} = \langle \mathbb D,\mathbb I(\mathbb D), r_1, r_2, \ldots,$ $r_n \rangle$, where $S = \{r_1, \ldots, r_n \} \subseteq \allr$.
\end{itemize}
\end{definition}

\noindent Since all relations between intervals, points, and all mixed relations are already implicit in $\mathbb I(\mathbb D)$, we shall often simply write $\langle \mathbb I(\mathbb D)\rangle$ for a concrete interval structure $\langle \mathbb I(\mathbb D), r_1, r_2, \ldots, r_n \rangle$, and $\langle \mathbb{D}, \mathbb I(\mathbb D)\rangle$ for a concrete point-interval structure $\langle \mathbb D,\mathbb I(\mathbb D), r_1, r_2, \ldots, r_n \rangle$; this is in accordance with the standard usage in much of the literature on interval temporal logics. Moreover, we denote by $FO+S$ the language of first-order logic without equality and relation symbols corresponding to the relations in $S$. Finally, $\mathcal{F}$ is further said to be \emph{of the class $\mathrm C$} ($\mathrm C\in\{\Lin,\Den,\Dis,\Unb\}$) when $\mathbb D$ belongs to the specific class of linear orders $\mathrm C$.

%
%
%\noindent  Each $r_i$ is one of Allen's relations (see Tab.~\ref{ii:relations}), and
%\item $\mathbb I(\mathbb D)=\{[a,b]\mid a,b\in \mathbb D, a<b\}$ is built over a linear order $\mathbb D$.
%\end{itemize}
%\noindent Similarly, for  of relations, :
%\begin{itemize}
%\item Each $r_i$ is a interval-interval, interval-point, or point-point relation (see Tab.~\ref{ii:relations} and Tab.~\ref{ip:relations}), and
%\item $\mathbb I(\mathbb D)=\{[a,b]\mid a,b\in \mathbb D, a<b\}$ is built over a linear order $\mathbb D$.
%
% As we have done above, we shall often simply write
%
%  and we denote by $FO+S$ the two-sorted language of first-order logic with relation symbols corresponding to the relations in $S$ (notice that sorted equality is included in the language only if it is a member of $S$). In the two-sorted context, we shall use different symbols for variables that are supposed to be interpreted over different sorts; in particular, $x_p,y_p,\ldots$ will denote point variables, $x_i,y_i,\ldots$ interval variables, and $x,y,\ldots$ will be used when we do not want to specify the sort, or when the sort is clear from the context (see Tab.~\ref{tab:conv}).
%\end{definition}
%
%

\subsection{(Un)definability and Truth Preserving Relations}

We describe here the most important tools that we use to classify the expressive power of our (sub-)languages.

\begin{definition}
Let $S \subseteq \allr$, and $\mathrm C$ a class of linear orders. We say that $FO + S$ {\em defines} $r\in \allr$ over $\mathrm C$, denoted by $FO+S\rightarrow_{\mathrm C} r$, if there exists an $FO+S$-formula $\varphi(x,y)$ such that $\varphi(x,y) \leftrightarrow r(x,y)$ is valid on the class of concrete point-interval structures of signature $(S \cup \{ r \})$ based on $\mathrm C$.
\end{definition}

\noindent By $FO+S\rightarrow r$ we denote the fact that $FO+S\rightarrow_{\Lin} r$ (and hence $FO+S\rightarrow_{\mathrm C} r$ for every $\mathrm C\in\{\Lin,\Den,\Dis,\Unb\}$). Obviously, $FO+S\rightarrow r$ for all $r \in S$.

\begin{definition}
Let $S,S'\subseteq \allr$ and $\mathrm{C}$ a class of linear orders. We say that $S$ is:
\begin{itemize}
\item {\em $S'$-complete over $\mathrm C$} (resp., {\em $S'$-incomplete over $\mathrm C$}) if and only if $FO+S\rightarrow_{\mathrm C} r$ for all $r\in S'$ (resp., $FO+S\not\rightarrow_{\mathrm C} r$ for some $r\in S'$), and
\item {\em minimally $S'$-complete over $\mathrm C$} (resp., {\em maximally $S'$-incomplete over $\mathrm C$}) if and only if it is $S'$-complete (resp., $S'$-incomplete) over $\mathrm C$, and every proper subset (resp., every proper superset) of $S$ is $S'$-incomplete (resp., $S'$-complete) over the same class.
\end{itemize}
\end{definition}

\noindent The notion of (minimally) $r$-completeness and (maximally) $r$-incompleteness over $\mathrm C$ is immediately deduced from the above one, by taking $S'=\{r\}$ and denoting the latter simply by $r$. Moreover, one can project the above definitions over some interesting subsets of $\allr$, such as $\alli,\allm$ or $\allp$, obtaining relative completeness and incompleteness.

%     {\em minimally $r$-complete over $\mathrm C$} (resp., {\em maximally $r$-incomplete over $\mathrm C$}) if and only if $FO+S\rightarrow_{\mathrm C} r$ (resp., $FO+S\not\rightarrow_{\mathrm C} r$), and, every proper subset (resp., every proper superset) of $S$ is $r$-incomplete (resp., $r$-complete) over the same class.
%

%\noindent

\medskip

Let $C' \subseteq C$ be two classes of linear orders. Notice that if $FO+S\rightarrow_{\mathrm C} r$ then $FO+S\rightarrow_{\mathrm C'} r$ and, contrapositively, that if $FO+S \not \rightarrow_{\mathrm C'} r$ then $FO+S \not \rightarrow_{\mathrm C} r$. So specifically, if $S$ is $S'$-complete over $\mathrm C$, then it is also $S'$-complete over $\mathrm C'$. Also, if $S$ is $S'$-incomplete over $\mathrm C'$, then it is also $S'$-incomplete over $\mathrm C$. Notice however, that minimality and maximality of complete and incomplete sets does not necessarily transfer between super and subclasses in a similar way. In what follows, in order to prove that $FO+S \not \rightarrow_{\mathrm C} r$ for some $r$ and some class $\mathrm C$, we shall repeatedly apply the following definition and (rather standard) procedure.

\begin{definition}\label{Def:S:Truth:Pres:Rel}
Let $\mathcal{F} = \langle \mathbb D,\mathbb I(\mathbb D),S\rangle$ and $\mathcal{F}' = \langle \mathbb D',\mathbb I(\mathbb D'), S\rangle$ be concrete structures where $S\subseteq\allr$. A binary relation $\zeta\subseteq(\mathbb D\cup\mathbb I(\mathbb D)) \times (\mathbb D' \cup \mathbb I(\mathbb D'))$ is called a \emph{surjective $S$-truth preserving relation} if and only if:

\medskip

\begin{enumerate}[label=\emph{(\roman*)}]
\item $\zeta$ respects sorts, i.e., $\zeta = \zeta_p \cup \zeta_i$, where $\zeta_p \subseteq \mathbb D \times \mathbb D'$ and $\zeta_i \subseteq \mathbb I(\mathbb D) \times \mathbb I(\mathbb D')$;
\item $\zeta$ respects the relations in $S$, i.e., if $(a,a'),(b,b')\in\zeta_p$ and $(I,I'),(J,J')\in\zeta_i$, then:
  \begin{enumerate}[noitemsep,topsep=0pt]
  \item $r(a,b)$ if and only if $r(a',b')$ for every point-point relation $r \in S$;
  \item $r(I,a)$ if and only if $r(I',a')$ for every interval-point relation $r\in S$;
  \item $r(I,J)$ if and only if $r(I',J')$ for every interval-interval relation $r \in S$;
  \end{enumerate}

\item $\zeta$ is total and surjective, i.e.:
  \begin{enumerate}[noitemsep,topsep=0pt]
    \item for every $a \in \mathbb D$ (resp., $I \in \mathbb I(\mathbb D)$), there exist $a' \in \mathbb D'$ (resp., $I' \in \mathbb I(\mathbb D')$) such that $(a,a')\in\zeta_p$ (resp., $(I,I')\in\zeta_i$);
    \item for every $a' \in \mathbb D'$ (resp., $I' \in \mathbb I(\mathbb D')$), there exist $a \in \mathbb D$ (resp., $I \in \mathbb I(\mathbb D)$) such that $(a,a')\in\zeta_p$ (resp., $(I,I')\in\zeta_i$).
    \end{enumerate}
\end{enumerate}
\end{definition}

\noindent If we add to Definition \ref{Def:S:Truth:Pres:Rel} the requirement that that $\zeta$ should be functional, we obtain nothing but the definition of an isomorphism between two-sorted first-order structures or, equivalently, an isomorphism between single sorted first-order structures with predicates added for `point' and `interval' (see e.g.\ \cite{hodges1993model}). As one would expect, surjective $S$-truth preserving relations preserve the truth of all first-order formulas in signature $S$. This is stated in Theorem \ref{theo:truth}, below. The reason why we consider only interval-point relations instead of all mixed relations is that, as we shall explain, we can limit ourselves to work without inverse relations, and point-interval relations are the inverse of interval-point ones.

\begin{definition}
If $\zeta$ is a surjective $S$-truth preserving relation, we say that $\zeta$ {\em breaks} $r\not\in S$ if and only if there are:
\begin{enumerate}[label=\emph{(\roman*)}]
    \item $(a,a'),(b,b')\in\zeta_p$ such that $r(a,b)$ but $\neg r(a',b')$, if $r$ is point-point, or
    \item $(a,a')\in\zeta_p$ and $(I,I')\in\zeta_i$ such that $r(I,a)$ but $\neg r(I',a')$, if $r$ is interval-point, or
    \item $(I,I'),(J,J')\in\zeta_i$ such that $r(I,J)$ but $\neg r(I',J')$, if $r$ is interval-interval.
\end{enumerate}
\end{definition}

\noindent The following result is, as already mentioned, a straightforward generalization of the classical result on the preservation of truth under isomorphism between first-order structures, and it is proved by an easy induction on formulas, using clause (ii) of Definition \ref{Def:S:Truth:Pres:Rel} to establish the base case for atomic formulas and clause (iii) for the inductive step for the quantifiers.

\begin{theorem}\label{theo:truth}%\marginpar{\tiny {\bf Willem: can you check reviewer's first comment in the 'summary', add a ref/proof, and check detailed comments referred to as 'pag.7, theo 6'?}}
If $\zeta=\zeta_p\cup\zeta_i$ is a surjective $S$-truth preserving relation between $\mathcal{F} = \langle \mathbb D,\mathbb I(\mathbb D),S\rangle$ and $\mathcal{F}' = \langle \mathbb D',\mathbb I(\mathbb D'),S\rangle$, and $a_1,\ldots,a_k \in \mathbb D$, $a'_1,\ldots, a'_k \in \mathbb D$, $I_1,\ldots,I_l \in \mathbb I(\mathbb D)$, and $I'_1,\ldots,I'_l \in \mathbb I(\mathbb D')$ are such that $(a_j,a_j')\in\zeta_p$ for $1 \leq j \leq k$, and $(I_j,I'_j)\in\zeta_i$ for $1 \leq j \leq l$, then for every $FO+S$ formulas $\varphi(x_p^1,\ldots, x_p^k,y_i^1,\ldots, y_i^l)$ with free variables $x_p^1,\ldots x_p^k, y_i^1, \ldots y_i^l$, we have that
\[
\mathcal{F} \models \varphi(a_1, \ldots, a_k, I_1, \ldots, I_l) \text{ \ if and only if \ } \mathcal{F}' \models \varphi(a'_1, \ldots a'_k, I'_1, \ldots, I'_l).
\]
\end{theorem}

\noindent Thus, to show that $FO + S \not \rightarrow r$ for a given $r\in\allr$, it is sufficient to find two concrete point-interval structures $\mathcal{F}$ and $\mathcal{F}'$ and a surjective $S$-truth preserving relation $\zeta$ between  $\mathcal{F}$ and $\mathcal{F}'$ which breaks $r$. For the readers' convenience, let us refer to surjective $S$-truth preserving relations as simply $S$-{\em relations}.

Although there are other constructions that could be used to show that relations are not definable in $FO + S$, e.g.\ elementary embedding or Ehrenfeucht-Fra\"{i}ss\'{e} games, we have found $S$-relations sufficient for our purposes in this paper.

\subsection{Strategy}

The main objective of this paper is to establish all expressively different subsets of $\allr$ (and, then, of $\alli,\allm$ or $\allp$) over the mentioned classes of linear orders. To this end, for each $r\in\allr$ we compute all expressively different minimally $r$-complete and all maximally $r$-incomplete subsets of $\allr$, from which we can easily deduce all expressively different minimally $r$-complete and maximally $r$-incomplete subsets of $\alli,\allm$ and $\allp$; minimally $\allr$- (resp., $\alli-,\allm-,\allp-$) complete and maximally incomplete subsets are, then, deduced as a consequence of the above results. The set $\allr$ contains, as we have mentioned, 26 different relations. This means that there are $2^{26}$ potentially different extensions of first-order logic to be studied. Clearly, unless we design a precise strategy that allows us to reduce the number of results to be proved, the task becomes cumbersome.

\medskip

As a first simplification principle observe that, since we are working within first-order logic, all inverses of relations are explicitly definable, and hence we only need to assume as primitive a set which contains all relation up to inverses, which implies that point-interval relations can be omitted if we consider all interval-point ones. Accordingly, let $\alli^{+}$  be the set of interval-interval relations given in Tab.~\ref{ii:relations} together with $=_i$, $\allm^{+}$ be the set of interval-point relations given in Tab.~\ref{ip:relations}, and let $\allp^{+}=\{<,=_p\}$. Lastly let $\allr^{+}=\alli^{+} \bigcup \allm^{+} \bigcup \allp^{+}$.

\medskip

In order to further reduce the number of results to be presented, consider what follows. The {\em order dual} of a structure $\mathcal{F} = \langle\mathbb D,\mathbb I(\mathbb D)\rangle$ is the structure $\mathcal{F}^{\partial} = \langle\mathbb D^\partial,\mathbb I(\mathbb D^{\partial})\rangle$  based on the order dual $\mathbb D^{\partial}$ (obtained by reversing the order) of the underlying linear order $\mathbb D$. All classes considered in this paper are closed under taking order duals.

\begin{definition}
The \emph{reversible relations} are exactly the members of the set $\{\ip 0,\ip1,$ $\ip3,\ip4,$ $\ii 14,\ii 03\}$. The relations belonging to the complement  $\allr^{+} \setminus \{\ip 0,\ip1,$ $\ip3,\ip4,\ii 14,\ii 03\}$ are called \emph{symmetric}; if, in addition, $r=\ip2$ or $r=\ii04$, then $r$ is said {\em self-symmetric}. If $r=\ip 0$ (resp., $r=\ip 1, r=\ii14$), its {\em reverse} is $r=\ip 4$ (resp., $r=\ip 3,r=\ii 03$), and the other way around. Finally, the {\em symmetric} $S'$ of a subset $S\subseteq\allr^{+}$ is obtained by replacing every reversible relation in $S$ with its reverse. We shall use the notation $S \sim S'$ to indicate that sets $S$ and $S'$ are symmetric.
\end{definition}

\begin{figure}[t]
\scriptsize
\centering
\begin{code}{50mm}
\FUNCTION{Undef}{r\in\allr^{+},def\_rules}
\BEGIN
\FORALL \ S\subset\allr^{+}\\
\BEGIN
S=Closure(S,def\_rules);\\
\IF ((r\notin S) \AND (S \ is \ maximal)) \THEN
list \ S
\END\\
\END \\
\end{code}
\begin{code}{60mm}
\FUNCTION{Closure}{S,def\_rules}
\BEGIN
\WHILE  (S \ changes)\\
\BEGIN
\FORALL 1\le i\le size(def\_rules)\\
\BEGIN
\IF (def\_rules[i] \ applies) \THEN S=Apply(S,def\_rule[i])\\
\END\\
\END\\
\RETURN\ S
\END \\
\end{code}
\caption{Pseudo-code to identify maximally $r$-incomplete sets.}\label{fig:undef}
\end{figure}

\noindent This definition is motivated by the following easily verifiable facts. Let $r \in \allr^+$, $\mathcal{F}$ be a structure, and $x$ and $y$ be elements of $\mathcal{F}$ of the appropriate sorts for $r$; then:
\begin{enumerate}[label=\emph{(\roman*)}]
\item if $r$ is a reversible relation, with reverse $r'$, then $\mathcal{F} \models r(x,y)$ if and only if $\mathcal{F}^\partial \models r'(x,y)$;
\item if $r$ is self-symmetric, then $\mathcal{F} \models r(x,y)$ if and only if $\mathcal F^{\partial} \models r(x,y)$;
\item if $r$ is a symmetric, but not self-symmetric, relation, then $\mathcal F \models r(x,y)$ if and only if $\mathcal F^{\partial} \models r(y,x)$.
\end{enumerate}

\begin{table}[t]
\begin{center}
\begin{tabular}{|p{0.22\textwidth}|p{0.60\textwidth}|}
\hline
                  $\mathbb D,\mathbb D',\ldots$            & (generic) linearly ordered sets     \\
                  $x_p,y_p,\ldots$                         & first-order variables for points \\
                  $x_i,y_i,\ldots$                         & first-order variables for intervals \\
                  $x,y,\ldots$                             & first-order variables of any sort \\
                  {\em before},\ldots                       & relations in text are {\em emphasized}\\
                  $\mathcal F,\mathcal F',\ldots$          & (generic) concrete (point-)interval structures     \\
                  $S,S',\ldots$                            & (generic) subsets of $\allr$-relations  \\
                  $\zeta~(\zeta_p,\zeta_i)$                & surjective relation (for points, for intervals) \\
                  $Id_p (Id_i)$                            & `identity' relation over points (intervals) \\
                  $\mathrm C,\mathrm C'$                   & (generic) class of linearly ordered sets\\
                  $FO+S\rightarrow_{\mathrm C} r$          & $S$ defines $r$ (w.r.t. the class $C$)\\
                  $S\sim S'$                               & $S$ and $S'$ are symmetric\\
                  $a\in\mathbb D$                          & $a$ is a point of $D$, where $\mathbb D=(D,<)$\\
                  $\underline S$                           & in the text, a new proof case is \underline{underlined}\\
                  $r$                                      & generic relation \\
                  $\mathsf{mcs}$ ($\mathsf{mcs}(r)$)                         & minimally complete set (minimally $r$-complete set)\\
                  $\mathsf{MIS}$ ($\mathsf{MIS}(r)$)                         & maximally incomplete set (maximally $r$-incomplete set)\\
\hline
\end{tabular}
\end{center}\caption{Notational conventions used in this paper.}\label{tab:conv}
\end{table}

\begin{table}[t]
\small
\begin{center}
\begin{tabular}{|l|l|l|l|l|l|}
%\begin{tabular}{|p{0.05\textwidth}|p{0.09\textwidth}|p{0.18\textwidth}|p{0.18\textwidth}|p{0.18\textwidth}|p{0.18\textwidth}|}
\hline
$=_p$        & $=_i$            & $<$                     & $\ip 0$                     & $\ip 1$                     & $\ip 2$                  \\
\hline
$\{<\}$      & $\{\ip 0,\ip 2\}$&$\{\ip0, \ip1 \}$        & $\{\ip1, < \}$              & $\{\ip0, \ip3 \}$           &$\{\ip0, \ip4 \}$          \\
$\{\ip 1\}$  & $\{\ip 0,\ip 3\}$&$\{\ip0, \ip3 \}$        & $\{\ip1, \ip2 \}$           & $\{\ip2, \ip3 \}$           &$\{\ip1, \ip3 \}$          \\
$\{\ip 3\}$  & $\{\ip 0,\ip 4\}$&$\{\ip1, \ip2 \}$        & $\{\ip1, \ip3 \}$           & $\{\ip2, \ip4, < \}$        &$\{\ip1, \ii03\}$          \\
             & $\{\ip 1,\ip 2\}$&$\{\ip1, \ip3 \}$        & $\{\ip1, \ip4 \}$           & $\{\ip2, \ii14, <\}$        &$\{\ip3, \ii14\}$          \\
             & $\{\ip 1,\ip 3\}$&$\{\ip3, \ip4 \}$        & $\{\ip2, \ip3 \}$           & $\{\ip2, \ii03, <\}$        &$\{\ip0, \ip3 \}$            \\
             & $\{\ip 1,\ip 4\}$&$\{\ip1, \ip4 \}$        & $\{\ip2, \ip4 \}$           & $\{\ip2, \ii34, <\}$        &$\{\ip0, \ii14, \ii24, < \}$ \\
             & $\{\ip 2,\ip 3\}$&$\{\ip2, \ip3 \}$        & $\{\ip2, \ii14, < \}$       & $\{\ip3, \ii14 \}$          &$\{\ip0, \ii14, \ii04, < \}$ \\
             & $\{\ip 2,\ip 4\}$&$\{\ip3, \ii14 \}$       & $\{\ip2, \ii03, < \}$       & $\{\ip3, \ii24, \ii03 \}$   &$\{\ip0, \ii14, \ii44, < \}$ \\
             & $\{\ii 14\}$     &$\{\ip3, \ii34 \}$       & $\{\ip3, \ii14 \}$          & $\{\ip3, \ii03, \ii04 \}$   &$\{\ip0, \ii03, < \}$        \\
             & $\{\ii 03\}$     &$\{\ip1, \ii03 \}$       & $\{\ip4, \ii14, < \}$       & $\{\ip3, \ii03, \ii44 \}$   &$\{\ip0, \ii34, < \}$        \\
             & $\{\ii 34\}$     &$\{\ip1, \ii34 \}$       & $\{\ip1, \ii14, \ii24 \}$   & $\{\ip3, \ii34 \}$          &$\{\ip1, \ip4 \}$            \\
             &                  &$\{\ip1, \ii04, =_i \}$  & $\{\ip1, \ii14, \ii04 \}$   & $\{\ip4, \ii14, < \}$       &$\{\ip1, \ii14, \ii24 \}$    \\
             &                  &$\{\ip1, \ii24, =_i \}$  & $\{\ip1, \ii14, \ii44 \}$   & $\{\ip4, \ii24, \ii03, < \}$&$\{\ip1, \ii14, \ii04 \}$    \\
             &                  &$\{\ip1, \ii44, =_i \}$  & $\{\ip1, \ii24, =_i \}$     & $\{\ip4, \ii03, \ii04, < \}$&$\{\ip1, \ii14, \ii44 \}$    \\
             &                  &$\{\ip3, \ii04, =_i \}$  & $\{\ip1, \ii04, =_i \}$     & $\{\ip4, \ii03, \ii44, < \}$&$\{\ip1, \ii34 \}$           \\
             &                  &$\{\ip3, \ii24, =_i \}$  & $\{\ip1, \ii44, =_i \}$     & $\{\ip4, \ii34, < \}$       &$\{\ip3, \ii24, \ii03 \}$    \\
             &                  &$\{\ip3, \ii44, =_i \}$  & $\{\ip1, \ii03 \}$          &                             &$\{\ip3, \ii03, \ii04 \}$    \\
             &                  &$\{\ip1, \ii14, \ii24 \}$& $\{\ip1, \ii34 \}$          &                             &$\{\ip3, \ii03, \ii44 \}$    \\
             &                  &$\{\ip1, \ii14, \ii04 \}$& $\{\ip2, \ii34, < \}$       &                             &$\{\ip3, \ii34 \}$           \\
             &                  &$\{\ip1, \ii14, \ii44 \}$& $\{\ip3, \ii34 \}$          &                             &$\{\ip4, \ii14, < \}$        \\
             &                  &$\{\ip3, \ii24, \ii03 \}$& $\{\ip3, \ii24, \ii03 \}$   &                             &$\{\ip4, \ii24, \ii03, < \}$ \\
             &                  &$\{\ip3, \ii03, \ii04 \}$& $\{\ip3, \ii03, \ii04 \}$   &                             &$\{\ip4, \ii03, \ii04, < \}$ \\
             &                  &$\{\ip3, \ii03, \ii44 \}$& $\{\ip3, \ii03, \ii44 \}$   &                             &$\{\ip4, \ii03, \ii44, < \}$ \\
             &                  &                         & $\{\ip4, \ii34, < \}$       &                             &                             \\
             &                  &                         & $\{\ip4, \ii24, \ii03, < \}$&                             &                             \\
             &                  &                         & $\{\ip4, \ii03, \ii04, < \}$&                             &                             \\
             &                  &                         & $\{\ip4, \ii03, \ii44, < \}$&                             &                             \\
\hline
\end{tabular}
\caption{The spectrum of the $\mathsf{mcs}(r)$, for each $r\in\allm^+\cup\{=_p,=_i,<\}$. - Class: $\Lin$ (review).}\label{tab:alllin1}
\end{center}
\end{table}

\noindent The following crucial lemma capitalizes on these facts.

\begin{lemma}\label{lem:symm}
Let $S, S' \subset\allr^{+}$ be such that $S \sim S'$. If $r$ is a symmetric relation, then $FO+S\rightarrow r$ if and only if  $FO+S'\rightarrow r$. Moreover, if $r$ is a reversible relation with reverse $r'$, then  $FO+S\rightarrow r$ if and only if  $FO+S'\rightarrow r'$.
\end{lemma}

\proof Let $S, S' \subset\allr^{+}$ such that $S \sim S'$. For any $FO+S$ formula $\varphi$ that defines a given relation (and, therefore, with exactly two free variables), let $\varphi'$ denote the formula obtained from $\varphi$ by replacing every occurrence of a reversible relation with its reverse, and by swapping the arguments of every symmetric, but not self-symmetric, relation (occurrences of every self-symmetric relation are left unchanged). Induction on formulas then shows that $\mathcal F \models \varphi(x,y)$ (after substituting $x,y$ with elements of the appropriate sorts) if and only if $\mathcal F^\partial \models \varphi'(x,y)$, for any structure $\mathcal F$. The base case of the induction is taken care of by the three observations preceding this lemma. Now, suppose that a $FO+S$ formula $\varphi(x,y)$ defines a symmetric relation $r$. We claim that $\varphi'$ also defines $r$. Let $\mathcal F$ be an arbitrary structure of signature $S\cup\{r\}$. Then $\mathcal F^{\partial} \models \varphi(x,y) \leftrightarrow r(x,y)$, and hence $\mathcal F \models \varphi'(x,y) \leftrightarrow r(y,x)$ if $r$ is not self-symmetric, and $\mathcal F \models \varphi'(x,y) \leftrightarrow r(x,y)$ otherwise. Next, suppose that the $FO+S$ formula $\varphi(x,y)$ defines a reversible relation $r$. We claim that $\varphi'$ defines its reverse $r'$. Let $\mathcal F$ be an arbitrary structure of signature $S\cup\{r\}$. Then $\mathcal F^{\partial} \models \varphi(x,y) \leftrightarrow r(x,y)$, and, hence, $\mathcal F \models \varphi'(x,y) \leftrightarrow r'(x,y)$.
\qed

In conclusion, we can limit our attention to 14 out of 26 relations by disregarding the inverses of relations in $\allr^{+}$, and we do not need to explicitly
analyze complete and incomplete sets for $\ip 3$, $\ip 4$, and $\ii 03$ as those correspond exactly to symmetric of complete and incomplete sets for $\ip 0$, $\ip 1$, and $\ii 14$, respectively. This means that only 11 relations are to be analyzed (which we can refer to as {\em explicit} relations).

\medskip

Even under the mentioned simplifications, there is a huge number of results to be presented and displayed. Let $r$ be anyone of the explicit relations.
In order to correctly identifying all minimally $r$-complete sets ($\mathsf{mcs}(r)$), we need to know all maximally $r$-incomplete sets ($\mathsf{MIS}(r)$) over the same class, and the other way around. To this end, we proceed in the following way:

\begin{enumerate}
\item fixed a class of linearly ordered sets and an explicit relation $r$, we first guess the $r$-complete subsets of $\allr^{+}$, obtaining
a first approximation of the definability rules for $r$ (here, denoted by $def\_rules$);
\item then, we apply the algorithm in Fig.~\ref{fig:undef}, which uses the set of $r$-complete subsets of $\allr^+$ (the parameter $def\_rules$) to obtain a first approximation of the maximally $r$-incomplete sets (the procedure $Closure()$ returns the transitive closure of a set of relations, obtained by systematically applying the definition rules contained in $def\_rules$ - function $Apply()$);
\item after that, we prove that every $R_1,R_2,\ldots,R_k$ listed as a maximally $r$-incomplete set is actually $r$-incomplete, and, if not, we repeat from step 1, using the acquired knowledge to update the set of $r$-complete subsets of $\allr^+$;
\item at this point, the sets $S_1,S_2,\ldots,S_{k'}$ listed at step 1 are, actually, all minimally $r$-complete. To see this, observe that, for each $i$, $S_i$ is $r$-complete by definition, and if there was a $r$-complete set $S\subset S_i$, then for some $R_j$ listed as maximally $r$-incomplete set we could not prove its $r$-incompleteness. Therefore, $S_1,S_2,\ldots,S_{k'}$ are, in fact, all minimally $r$-complete, and, as a consequence, $R_1,R_2,\ldots,R_k$ are all maximally $r$-incomplete.
\end{enumerate}

\noindent For example, suppose that we are interested in minimal $=_i$-complete (resp., maximal $=_i$-incomplete) sets of relations. First, we guess that having the relations $\ip 1,\ip2,$ and $\ip 3$ is sufficient to define equality between intervals, because we can easily define the latter by asserting that two intervals are the same if they {\em start} with the same point, {\em end} with the same point, and {\em contain} the same points; so, $S=\{\ip 1,\ip2,\ip 3\}$ is our first guess. Now, by computing all possible subsets, we find that if $S$ were minimal, then $\allr^{+}\setminus S$ should be maximally $=_i$-incomplete. Since we cannot find a proof of the latter, we realize that, in fact, $S'=\{\ip 1,\ip 3\}$ is sufficient to define $=_i$. This back-and-forth procedure continues until it stabilizes; clearly, one cannot find one minimal $r$-complete subset without finding, at the same time, all maximal $r$-incomplete subsets, and the other way around. In the following, for each completeness result, its minimality is left unproven; after the corresponding undefinability results are proven, then, by observing that our sets are systematically computed, their respective minimality and maximality will be a consequence.

%Once the above process is completed for every relation, we can then easily deduce all minimally $\allr^{+}$-complete and all maximally
%$\allr^{+}$-incomplete sets, to complete the picture. A similar procedure works for $\alli^{+}, \allm^{+}$, and $\allp^{+}$.

\medskip

The most common notational conventions used in the paper are listed in Tab.~\ref{tab:conv}.

\section{Completeness Results in the Class $\Den$}\label{sec:den}

\begin{table}[t]
\small
\begin{center}
\begin{tabular}{|p{0.17\textwidth}|p{0.17\textwidth}|p{0.17\textwidth}|p{0.17\textwidth}|p{0.17\textwidth}|}
\hline
$\ii 34$            & $\ii 14$          & $\ii 24$                & $\ii 04$                 & $\ii 44$                   \\
\hline
$\{\ip 1,\ip 3\}$   & $\{\ip 0,\ip2 \}$ &$\{\ip 2,<\}$            & $\{\ip 0,\ii 24,\ii44\}$ & $\{\ip 0,\ii 24,\ii04\}$    \\
$\{\ip 2,\ii 14 \}$ & $\{\ip 0,\ip4 \}$ &$\{\ip 0,\ii 44,\ii04\}$ & $\{\ip 1,\ii 24,\ii44\}$ & $\{\ip 1,\ii 24,\ii04\}$    \\
$\{\ip 4,\ii 14 \}$ & $\{\ip0, \ip3 \}$ &$\{\ip 1,\ii 44,\ii04\}$ & $\{\ip 2,\ii 24,\ii44\}$ & $\{\ip 4,\ii 24,\ii04\}$    \\
$\{\ip0, \ip2 \}$   & $\{\ip0, \ii03\}$ &$\{\ip 4,\ii 44,\ii04\}$ & $\{\ip 4,\ii 24,\ii44\}$ & $\{\ip 3,\ii 24,\ii04\}$    \\
$\{\ip2, \ii03 \}$  & $\{\ip1, \ip2 \}$ &$\{\ip 3,\ii 44,\ii04\}$ & $\{\ip 3,\ii 24,\ii44\}$ & $\{\ip0, \ip2 \}$             \\
$\{\ip0, \ii03 \}$  & $\{\ip1, \ip3 \}$ &$\{\ip0,\ip2 \}$         & $\{\ip0, \ip2 \}$        & $\{\ip0, \ip3 \}$           \\
$\{\ip2, \ip4 \} $  & $\{\ip1, \ip4 \}$ &$\{\ip0,\ip3 \}$         & $\{\ip0, \ip3 \}$        & $\{\ip0, \ip4 \}$           \\
$\{\ii 24,\ii 14 \}$& $\{\ip1, \ii03\}$ &$\{\ip0,\ip4 \}$         & $\{\ip0, \ip4 \}$        & $\{\ip0, \ii03\}$           \\
$\{\ii 14,\ii 44 \}$& $\{\ip2, \ip3 \}$ &$\{\ip0,\ii03\}$         & $\{\ip0, \ii03\}$        & $\{\ip1, \ip2\}$               \\
$\{\ii 14,\ii 03 \}$& $\{\ip2, \ip4 \}$ &$\{\ip1,\ip2 \}$         & $\{\ip1, \ip2 \}$        & $\{\ip1, \ip3 \}$              \\
$\{\ii 14, \ii04 \}$& $\{\ip2, \ii03\}$ &$\{\ip1,\ip3 \}$         & $\{\ip1, \ip3 \}$        & $\{\ip1, \ip4 \}$              \\
$\{\ii24, \ii03 \}$ & $\{\ii24, \ii03\}$&$\{\ip1,\ip4 \}$         & $\{\ip1, \ip4 \}$        & $\{\ip1, \ii03 \}$             \\
$\{\ii03, \ii04 \}$ & $\{\ii03, \ii04\}$&$\{\ip1,\ii03\}$         & $\{\ip1, \ii03 \}$       & $\{\ip2, \ip3 \}$              \\
$\{\ii03, \ii44 \}$ & $\{\ii03, \ii44\}$&$\{\ip2,\ip3 \}$         & $\{\ip2, \ip3 \}$        & $\{\ip2, \ip4 \}$              \\
$\{\ip0, \ip3 \}$   & $\{\ii34 \}$      &$\{\ip2,\ip4 \}$         & $\{\ip2, \ip4 \}$        & $\{\ip2, \ii14 \}$             \\
$\{\ip0, \ip4 \}$   &                   &$\{\ip2,\ii14\}$         & $\{\ip2, \ii14 \}$       & $\{\ip2, \ii03 \}$             \\
$\{\ip1, \ip2 \}$   &                   &$\{\ip2,\ii03\}$         & $\{\ip2, \ii03 \}$       & $\{\ip3, \ii14 \}$             \\
$\{\ip1, \ip4 \}$   &                   &$\{\ip3,\ii14\}$         & $\{\ip2, \ii44,< \}$     & $\{\ip4, \ii14 \}$             \\
$\{\ip1, \ii03 \}$  &                   &$\{\ip4,\ii14 \}$        & $\{\ip3, \ii14 \}$       & $\{\ii14, \ii24 \}$            \\
$\{\ip2, \ip3 \}$   &                   &$\{\ii14,\ii03\}$        & $\{\ip4, \ii14 \}$       & $\{\ii14, \ii03 \}$            \\
$\{\ip3, \ii14 \}$  &                   &$\{\ii14,\ii04\}$        & $\{\ii14, \ii24 \}$      & $\{\ii14, \ii04 \}$          \\
                    &                   &$\{\ii14,\ii44\}$        & $\{\ii14, \ii03 \}$      & $\{\ii24, \ii03 \}$          \\
                    &                   &$\{\ii03,\ii44\}$        & $\{\ii14, \ii44 \}$      & $\{\ii03, \ii04 \}$          \\
                    &                   &$\{\ii04,\ii 03\}$       & $\{\ii24, \ii03 \}$      & $\{\ii34 \}$                    \\
                    &                   &$\{\ii34\}$              & $\{\ii03, \ii44 \}$      &                                 \\
                    &                   &                         & $\{\ii34 \}$             &                                 \\
\hline
\end{tabular}
\caption{The spectrum of the $\mathsf{mcs}(r)$, for each $r\in\alli^+\setminus\{=_i\}$. - Class: $\Lin$ (review).}\label{tab:alllin2}
\end{center}
\end{table}

As it turns out, only a few of the results that appear in Tab.~\ref{tab:alllin1} and
Tab.~\ref{tab:alllin2} are to be refined to obtain all minimally complete sets under the assumption of density. In the following, we show in boldface those complete sets that are stricter than some complete set for a given relation, or completely new, and we prove them explicitly modulo symmetry.

\subsection{Definability in  $\alli^+$ and in $\allm^+$.}\label{subsec:allimlin}

We begin our study of definability over the Class $\Den$ by considering minimal definability of $\alli^+$ relations by sets of $\alli^+$ relations, and then of $\allm^+$ relations by sets of $\allm^+$ relations.

\begin{table}[t]
	\small
	\begin{center}
		\begin{tabular}{|p{0.14\textwidth}|p{0.14\textwidth}|p{0.14\textwidth}|p{0.14\textwidth}|p{0.14\textwidth}|p{0.14\textwidth}|}
			\hline
			$\ii 34$           &$\ii 14$           &$\ii 24$             &$\ii 04$            &$\ii 44$ 		     &$=_i$\\
			\hline
			$\{\ii14,\ii03\}$   &$\{\ii34\}$  &$\{\ii34\}$		     &$\{\ii34\}$		  &$\{\ii34\}$ 		&$\{\ii34\}$\\
			$\{\ii14,\ii04\}$          &$\{\ii14,\ii04\}$  &${\{\ii44\}}$	 &$\boldsymbol{\{\ii24\}}$   &$\boldsymbol{\{\ii24\}}$&$\{\ii 14\}$\\
			$\boldsymbol{\{\ii24\}}$          &$\boldsymbol{\{\ii24\}}$  &$\{\ii14,\ii03\}$    &$\boldsymbol{\{\ii44\}}$   &$\{\ii14,\ii03\}$&$\{\ii 03\}$\\
			$\{\ii03,\ii04\}$          &$\{\ii03,\ii04\}$  &$\{\ii14,\ii04\}$    &$\{\ii14,\ii03\}$   &$\{\ii14,\ii04\}$&\\
			$\boldsymbol{\{\ii44\}}$          &$\boldsymbol{\{\ii44\}}$  &$\{\ii03,\ii04\}$    &                    &$\{\ii03,\ii04\}$&\\
			\hline
			
		\end{tabular}
		\caption{The spectrum of the $\mathsf{mcs}(r)$, for each $r\in\alli^+$. - Class: $\Den$.}\label{tab:mcsIntIntDen} (Lemma~\ref{lem:tab:mcsIntIntDen}) %(left) and of the $\mathsf{mcs}(\allm^+)$ (right). - Class: $\Lin$.}\label{tab:allilin}
	\end{center}
\end{table}

\begin{lemma}\label{lem:tab:mcsIntIntDen}
	Tab.~\ref{tab:mcsIntIntDen} is correct.
\end{lemma}

\begin{proof}
	We begin by proving that \underline{$\{\ii 24\}$} is $\ii 34$-complete. In Part I that we proved that on linear domains the relation $\ii 34 \vee \ii 44$ is $\ii 34$-complete. It is therefore sufficient to show that $\{\ii 24\}$ is $(\ii 34 \vee \ii 44)$-complete over dense structures. To this end, consider the following definition:
	
	\medskip
	
	\begin{small}
		\[
		\begin{array}{llll}
		x_i\ii 34\vee\ii 44 y_i &  \leftrightarrow  & \neg(x_i\ii 24 y_i)\wedge\exists z_i(x_i\ii 24 z_i\wedge z_i\ii 24 y_i). & \{\ii 24\}
		\end{array}
		\]
	\end{small}
	
	\medskip
	
	\noindent  If $\mathcal{F} \models \varphi([a,b],[c,d])$, there must be some interval $z_i$ {\em overlapped by} $x_i$ and {\em overlapping} $y_i$, which implies $a<c$ and $b<d$ and since $x_i$ cannot {\em overlap} $y_i$, we obtain $b \leq c$ as required. Conversely, assume that $[a,b]\ii 34\vee\ii 44[c,d]$. Then using the density assumption we can take $z_i=[e,f]$, where $a<e<b$ and $c<f<d$, to witness $\varphi$. Next, we prove that \underline{$\{\ii 44\}$} is $\ii 14$-complete. Consider the following definition:
	
	\medskip
	
	\begin{small}
		\[
		\begin{array}{llll}
		x_i \ii 14 y_i & \leftrightarrow & \forall z_i(z_i\ii 44 x_i\leftrightarrow z_i\ii 44 y_i)\wedge\exists z_i(x_i\ii 44 z_i \wedge\neg(y_i\ii 44 z_i)). & \{\ii 44\}
		\end{array}
		\]
	\end{small}

\noindent First suppose that $\mathcal{F} \models \varphi([a,b],[c,d])$. If $a<c$, then we can find a point $e$ with $a<e<c$ and the first conjunct of $\varphi$ fails, witnessed by the interval $z_i=[a,e]$. Similarly $a>c$ also gives a contradiction, so we obtain $a=c$. The second conjunct of $\varphi$ gives an interval $z_i=[e,f]$ such that $b<e$ and $e \leq d$, and so we have $b<d$. On the other hand, if we assume that $a=c<b<d$, since the underlying linear order is dense, there exists $e$ such that $b<e<d$; by taking $z_i=[e,d]$, we witness the second conjunct of $\varphi$, while the first conjunct is satisfied trivially. The $\ii 34$-completeness of \underline{$\{\ii 44\}$} now follows from the just-proved $\ii 14$-completeness of $\{\ii 44\}$ and the $\ii 34$-completeness of $\{\ii 14, \ii 44\}$ over all linear orders (see Table \ref{tab:alllin2}). Having established this, all the remaining new completeness results listed in Table~\ref{tab:mcsIntIntDen} follows via the completeness of $\{\ii 34\}$ with respect to every $\alli^+$-relation.
\end{proof}

\begin{table}[t]
	\begin{tabular}{|l|l|l|}
			\hline
			$\ip 0$           & $\ip 1$                       & $\ip 2$           \\ \hline
			$\{\ip1, \ip2 \}$ & $\boldsymbol{\{\ip0\}}$       & $\{\ip0, \ip3 \}$ \\
			$\{\ip1, \ip3 \}$ & $\{\ip2, \ip3 \}$             & $\{\ip0, \ip4 \}$ \\
			$\{\ip2, \ip4 \}$ & $\boldsymbol{\{\ip2, \ip4\}}$ & $\{\ip1, \ip3 \}$ \\
			$\{\ip2, \ip3 \}$ &                               & $\{\ip1, \ip4 \}$ \\
			$\{\ip2, \ip4 \}$ &                               &                   \\
			&                               &                   \\ \hline
		\end{tabular}
	\caption{The spectrum of the $\mathsf{mcs}(r)$, for each $r\in\allm^+$. - Class: $\Den$.}\label{tab:mcsMixedMixedDen} (Lemma~\ref{lem:mcsMixedMixedDen})
\end{table}

The next two lemmas involve both interval-interval and interval-point relations, and therefore contrasts with the other results of the present subsection that all concern definability purely among interval-interval relations and purely among interval-point relations. However, anticipating these particular results here, which do not refer to any table in particular, will be very useful in order to simplify some definitions and chains of deduction further on.

\begin{lemma}\label{lem:Anticipate:Den}
	The set $\{\ip 0\}$ is $<$-complete over dense linear orders.
\end{lemma}

\begin{proof}
	Consider the following definition:

\medskip

\begin{small}
	\[ x_p < y_p \leftrightarrow \begin{array}{llll}
	(\neg\exists x_i(x_i\ip 0 y_p)\wedge\exists y_i(y_i\ip 0 x_p))\vee\exists x_i(x_i\ip 0 x_p\wedge \neg(x_i\ip 0 y_p)) & \{\ip 0\}\\
	%                                  & \\
	%                                   (\exists x_i(x_i\ip1 x_p)\wedge\neg\exists y_i(y_i\ip 1 y_p))\vee(\exists x_iy_i(x_i\ip1 x_p\wedge y_i\ip 1 y_p\wedge y_i\ii 04 x_i)) & \{\ip 1,\ii 04\}\\
	%                                   & \\
	%                                   (\exists x_i(x_i\ip1 x_p)\wedge\neg\exists y_i(y_i\ip 1 y_p))\vee (\exists x_iy_i(x_i\ip1 x_p\wedge y_i\ip 1 y_p\wedge x_i\ii 24 y_i)) & \{\ip 1,\ii 24\}\\
	%                                   & \\
	%                                   (\exists x_i(x_i\ip1 x_p)\wedge\neg\exists y_i(y_i\ip 1 y_p))\vee(\exists x_iy_i(x_i\ip1 x_p\wedge y_i\ip 1 y_p\wedge x_i\ii 44 y_i)). & \{\ip 1,\ii 44\}\\
	\end{array} %\right.
	\]
\end{small}

\medskip

\noindent Suppose that $\mathcal F\models\varphi(x_p,y_p)$. If the first disjunct of $\varphi$ holds then $y_p$ is the greatest point of the model, while $x_p$ is not, which gives $x_p<y_p$. Assume the second disjunct of $\varphi$ holds. Then the interval $x_i=[a,b]$ is such that $x_p=c<a$, but $y_p\ge a$, that is $x_p<y_p$. Conversely, suppose that $x_p=a,y_p=b$, and that $a<b$. As the structure is dense, if $b$ is the last point of the model, we have no interval {\em starting} at $b$ but infinitely many intervals $x_i$ such that $x_i\ip 0 a$, and if $b$ is not the last point of the model, then every $x_i$ {\em starting} at $b$ satisfies the second part of the definition.
\end{proof}

\begin{lemma}\label{lem:anticipate:mixed:dense}
	The sets $\{\ip 3, \ii 04\}$ and $\{\ip 4, \ii 04\}$ are $\ii 34$-complete over dense linear orders..
\end{lemma}
	
\begin{proof}
We will show that both sets $\{\ip 3, \ii 04\}$ and $\{\ip 4, \ii 04\}$  define $\ii 03$ in the dense case; then, since  $\{\ii 03, \ii 04\}$ defines $\ii 34$ in the general case, we have the result. Indeed, consider the following definitions:

\medskip
	
	\begin{small}
	  \[ x_i \ii 03 y_i\leftrightarrow \left\{	
       \begin{array}{ll}
		\exists z_p (x_i \ip 3 z_p \wedge y_i \ip 3 z_p) \wedge \exists z_i (z_i \ii 04 y_i \wedge \neg (z_i \ii 04 x_i)). &\{\ip 3, \ii 04\}\\
		&\\
		\forall z_p (x_i \ip 4 z_p \leftrightarrow y_i \ip 4 z_p) \wedge \exists z_i (z_i \ii 04 y_i \wedge \neg (z_i \ii 04 x_i)). &\{\ip 4, \ii 04\}\\
		\end{array}
\right.
		\]
	\end{small}

\medskip

\noindent Let us consider \underline{$\{\ip 3, \ii 04\}$}, first. If $[a,b] \ii 03 [c,d]$ then $c < a$ and $b =d$, so, taking $z_p$ equal to $b$ satisfies the first conjunct. By density there is a point $e$ such that $c < e < a$, and therefore taking $z_i$ equal to $[e,a]$ satisfies the second conjunct. Conversely, suppose that $\varphi([a,b],[c,d])$. The first conjunct of $\varphi$ ensures that $b = d$. If it were the case that $a \leq c$, every interval contained in $[a,b]$ would be contained in $[c,d]$, violating the second conjunct, so it can only be the case that $c < a$. We therefore have that $[a,b] \ii 03 [c,d]$, as desired. The argument for \underline{$\{\ip 4, \ii 04\}$} is similar, using the fact that $\forall z_p (x_i \ip 4 z_p \leftrightarrow y_i \ip 4 z_p)$ is true if and only if the intervals assigned $x_i$ and $y_i$ have the same end point.
\end{proof}

\begin{lemma}\label{lem:mcsMixedMixedDen}
	Tab.~\ref{tab:mcsMixedMixedDen} is correct.
\end{lemma}

\begin{proof}
First, we show that \underline{$\{\ip 0\}$} is $\ip 1$-complete. From Lemma~\ref{lem:Anticipate:Den} we know that $\{\ip 0\}$ defines ${<}$ in the dense case, so we may consider the  following definition:
	
	\begin{small}
		\[
		\begin{array}{llll}
		x_i \ip 1y_p & \leftrightarrow & \forall z_p (x_i \ip 0 z_p \leftrightarrow z_p < y_p). & \{\ip 0 \}
		\end{array}
		\]
	\end{small}
	
\noindent If $[a,b] \ip 1 c$, then $a = c$, and hence $\varphi([a,b],c)$. Conversely, suppose that $\varphi([a,b],c)$. So, if $c < a$, then $[a,b] \ip 0 c$ by definition of $\ip 0$ and so, by $\varphi([a,b],c)$ we obtain the contradiction $c < c$. Now, if $a < c$ then, by density, there is a point $d$ such that $a < d < c$. Hence  $\neg [a,b] \ip 0 d$ by the definition of $\ip 0$, but then $d \not < c$ by $\varphi([a,b],c)$ --- a contradiction. The only possible case is therefore that $a = c$, i.e., that $[a,b] \ip 1 c$. To see that \underline{$\{\ip2, \ip4 \}$} is $\ip 1$-complete, it suffices to note that $\{\ip2, \ip4 \}$ defines $\ip 0$ in the general case, and then appeal to the just proven fact that  $\{\ip 0\}$ defines $\ip 1$.
\end{proof}

\subsection{Definability for the Relations $=_p,=_i,<$ in $\Den$}

\begin{table}[t]
\small
\begin{center}
\begin{tabular}{|p{0.20\textwidth}|p{0.20\textwidth}|p{0.20\textwidth}|}
\hline
$=_p$                     & $=_i$                    & $<$                     \\
\hline
$\{<\}$                   & $\{\ip 0,\ip 3\}$        &$\boldsymbol{\{\ip 0\}}$     \\
$\boldsymbol{\{\ip 0\}}$  & $\{\ip 0,\ip 4\}$        &$\{\ip1, \ip2 \}$               \\
$\{\ip 1\}$               & $\{\ip 1,\ip 3\}$        &$\{\ip1, \ip3 \}$               \\
$\{\ip 3\}$               & $\{\ip 1,\ip 4\}$        &$\{\ip2, \ip3 \}$               \\
$\boldsymbol{\{\ip 4\}}$  & $\boldsymbol{\{\ip 2\}}$ &$\{\ip3, \ii14 \}$               \\
                          & $\{\ii 14\}$             &$\{\ip3, \ii34 \}$               \\
                          & $\{\ii 03\}$             &$\{\ip1, \ii03 \}$               \\
                          & $\{\ii 34\}$             &$\{\ip1, \ii34 \}$               \\
                          & $\{\ii 34\}$             &$\boldsymbol{\{\ip1, \ii04 \}}$  \\
                          & $\boldsymbol{\{\ii24\}}$ &$\boldsymbol{\{\ip1, \ii24 \}}$  \\
                          & $\boldsymbol{\{\ii04\}}$ &$\boldsymbol{\{\ip1, \ii44 \}}$  \\
                          & $\boldsymbol{\{\ii44\}}$ &$\boldsymbol{\{\ip3, \ii04 \}}$    \\
                          &                          &$\boldsymbol{\{\ip3, \ii24 \}}$      \\
                          &                          &$\boldsymbol{\{\ip3, \ii44 \}}$      \\
                          &                          &$\boldsymbol{\{\ip 4\}}$             \\
\hline
\end{tabular}
\caption{The spectrum of the $\mathsf{mcs}(r)$, for each $r\in\{=_p,=_i,<\}$. - Class: $\Den$.}\label{tab:EqPEqILtDen}
\end{center}
\end{table}

\begin{lemma}\label{lem:tab:EqPEqILtDen}
Tab.~\ref{tab:EqPEqILtDen} is correct.
\end{lemma}

\begin{proof}
We begin by showing that \underline{$\{\ip 0 \}$} is $=_p$-compete. Consider the following definition:

\begin{small}
\[
\begin{array}{llll}
x_p =_p y_p &\leftrightarrow& \forall x_i(x_i\ip 0 x_p\leftrightarrow x_i\ip 0 y_p) & \{\ip 0\}
\end{array}
 \]
\end{small}

\noindent Suppose, first, that $x_p=_p y_p$; we have that, either there exists no interval in the model, in which case the right-hand side of the definition is vacuously true, or there are infinitely many intervals (as the underlying linear order is dense), and every interval $x_i$ clearly has both $x_p$ and $y_p$ in relation $\ip 0$ with it or none of them, again satisfying the right-hand side. If, on the other hand, $\varphi(x_p,y_p)$ is satisfied, then, if there is no interval, then $x_p$ and $y_p$ coincide, and if there are infinitely many intervals, the only way to guarantee that the intervals that see both $x_p$ and $y_p$ via $\ip 0$ are exactly the same is to assign the same point to $x_p$ and $y_p$, as we wanted to prove. The $=_p$-completeness of \underline{$\{\ip 4 \}$} follows by symmetry (see Lemma \ref{lem:symm}). Now, consider $=_i$-completeness, and the following definitions:

\medskip

\begin{small}
\[ x_i =_i y_i \leftrightarrow \left\{ \begin{array}{ll}
                                   \forall x_p(x_i\ip 2 x_p\leftrightarrow y_i\ip 2 x_p) & \{\ip 2\}\\
                                   & \\
                                   \forall z_i(x_i\ii 24 z_i\leftrightarrow y_i\ii 24 z_i)\wedge \forall z_i(z_i\ii 24 x_i\leftrightarrow z_i\ii 24 y_i) & \{\ii 24\} \\
                                   & \\
                                   \forall z_i(x_i\ii 04 z_i\leftrightarrow y_i\ii 04 z_i)\wedge \forall z_i(z_i\ii 04 x_i\leftrightarrow z_i\ii 04 y_i) & \{\ii 04\} \\
                                   & \\
                                   \forall z_i(x_i\ii 44 z_i\leftrightarrow y_i\ii 44 z_i)\wedge \forall z_i(z_i\ii 44 x_i\leftrightarrow z_i\ii 44 y_i). & \{\ii 44\} \\
   \end{array} \right. \]
\end{small}

\medskip

\noindent As for $\underline{\{\ip 2\}}$, observe that two intervals over a dense linear order are equal if and only if they have the same internal points. Consider, now, the set \underline{$\{\ii 24\}$}. It is clear that if $x_i =_i y_i$ then $\varphi(x_i, y_i)$ holds. For the converse, suppose that $[a,b] \neq [c,d]$. Then either $[a,b]$ has an internal point which is not an internal point of $[c,d]$ or $[c,d]$ has an internal point which is not an internal point of $[a,b]$. In both cases either the first or the second conjunct of $\varphi$ does not hold. The remaining cases are treated with similar arguments. Let us now focus on the relation $<$ and the corresponding definitions:

\medskip

\begin{small}
\[ x_p < y_p \leftrightarrow \begin{array}{llll}
                                   \exists x_i(x_i\ip0 x_p\wedge \neg(x_i\ip0 y_p))& \{\ip 0\}\\
 %                                  & \\
%                                   (\exists x_i(x_i\ip1 x_p)\wedge\neg\exists y_i(y_i\ip 1 y_p))\vee(\exists x_iy_i(x_i\ip1 x_p\wedge y_i\ip 1 y_p\wedge y_i\ii 04 x_i)) & \{\ip 1,\ii 04\}\\
%                                   & \\
%                                   (\exists x_i(x_i\ip1 x_p)\wedge\neg\exists y_i(y_i\ip 1 y_p))\vee (\exists x_iy_i(x_i\ip1 x_p\wedge y_i\ip 1 y_p\wedge x_i\ii 24 y_i)) & \{\ip 1,\ii 24\}\\
%                                   & \\
%                                   (\exists x_i(x_i\ip1 x_p)\wedge\neg\exists y_i(y_i\ip 1 y_p))\vee(\exists x_iy_i(x_i\ip1 x_p\wedge y_i\ip 1 y_p\wedge x_i\ii 44 y_i)). & \{\ip 1,\ii 44\}\\
\end{array} %\right.
\]
\end{small}

\medskip

\noindent Consider the set \underline{$\{\ip 0\}$}. Suppose that $\mathcal F\models\varphi(x_p,y_p)$. Then, there must be an interval that {\em starts} after $x_p$ but not after $y_p$; if $x_p\ge y_p$, then it is impossible to find such an interval, so $x_p<y_p$. Conversely, suppose that $x_p=a,y_p=b$, and that $a<b$. As the structure is dense, there exists a point $c$ such that $a<c<b$, and therefore the interval $[c,b]$ can be used as a witness of $x_i$. All other new definitions follow directly from the results for the class $\Lin$ combined with the density hypothesis. In particular, to see that \underline{$\{\ip 1,\ii 04\}$} is $<$-complete, note that $\{\ii 04\}$ defines $=_i$ in the dense case and that $\{\ip 1,\ii 04,=_i\}$ defines $<$ in the general one. The case for \underline{$\{\ip 1,\ii 24\}$} follows since $\{\ii 24\}$ defines $=_i$ in the dense case and since $\{\ip 1,\ii 24,=_i\}$ defines $<$. The case \underline{$\{\ip 1,\ii 44\}$} follows since $\{\ii 44\}$ defines $=_i$ in the dense case and $\{\ip 1,\ii 44,=_i\}$ defines $<$. The case \underline{$\{\ip 3,\ii 04\}$} follows since, in the dense case, we have that $\{\ii 04\}$ defines $=_i$ and that $\{\ip 3,\ii 04,=_i\}$ defines $<$. The case \underline{$\{\ip 3,\ii 24\}$} follows since, under the density hypothesis, we have that  $\{\ii 24\}$ defines $=_i$ and $\{\ip 3,\ii 24,=_i\}$ defines $<$. The case \underline{$\{\ip 3,\ii 44\}$} follows since $\{\ii 44\}$ defines $=_i$ in the dense case and since we already know that $\{\ip 3,\ii 44,=_i\}$ defines $<$. Lastly, the fact that \underline{$\{\ip 4\}$} is $<$-complete follows from $\{\ip 0\}$ being able to define $<$ and from Lemma \ref{lem:symm}.
\end{proof}

\subsection{Definability for $\allm^+$-Relations in $\Den$}\label{sec:allM:Den}

To continue studying how the density hypothesis influences the ability of the sub-languages to express our relations, we focus now on mixed relations.

\begin{table}[t]
	\small
	\begin{center}
		\begin{tabular}{|p{0.20\textwidth}|p{0.20\textwidth}|p{0.20\textwidth}|}
			\hline
			$\ip 0$                          & $\ip 1$                          & $\ip 2$                          \\ \hline
			$\{\ip1, < \}$                   & $\boldsymbol{\{\ip0\}}$          & $\{\ip0, \ip3 \}$                \\
			$\boldsymbol{\{\ip2, < \}}$      & $\{\ip2, \ip3 \}$                & $\{\ip0, \ip4 \}$                \\
			$\{\ip1, \ip2 \}$                & $\boldsymbol{\{\ip2, \ip4\}}$    & $\boldsymbol{\{\ip0,  \ii24 \}}$ \\
			$\{\ip1, \ip3 \}$                & $\boldsymbol{\{\ip2,<\}}$        & $\boldsymbol{\{\ip0,  \ii44 \}}$ \\
			$\{\ip1, \ip4 \}$                & $\{\ip3, \ii14 \}$               & $\boldsymbol{\{\ip0, \ii03 \}}$  \\
			$\{\ip2, \ip3 \}$                & $\boldsymbol{\{\ip3, \ii24\}}$   & $\boldsymbol{\{\ip0, \ii34 \}}$  \\
			$\{\ip2, \ip4 \}$                & $\boldsymbol{\{\ip3, \ii04\}}$   & $\boldsymbol{\{\ip0, \ii04 \}}$  \\
			$\{\ip3, \ii14 \}$               & $\boldsymbol{\{\ip3, \ii44\}}$   & $\{\ip1, \ip3 \}$                \\
			$\boldsymbol{\{\ip4, \ii14 \}}$  & $\{\ip3, \ii34 \}$               & $\{\ip1, \ip4 \}$                \\
			$\boldsymbol{\{\ip1, \ii24 \}}$  & $\boldsymbol{\{\ip4, \ii14 \}}$  & $\boldsymbol{\{\ip1,  \ii24 \}}$ \\
			$\boldsymbol{\{\ip1, \ii04 \}}$  & $\boldsymbol{\{\ip4, \ii24, \}}$ & $\boldsymbol{\{\ip1,  \ii44 \}}$ \\
			$\boldsymbol{\{\ip1, \ii44 \}}$  & $\boldsymbol{\{\ip4,  \ii04\}}$  & $\{\ip1, \ii03 \}$  \\
			$\{\ip1, \ii03 \}$               & $\boldsymbol{\{\ip4, \ii44\}}$   & $\{\ip1, \ii34 \}$               \\
			$\{\ip1, \ii34 \}$               & $\boldsymbol{\{\ip4, \ii34\}}$   & $\boldsymbol{\{\ip1, \ii04 \}}$  \\
			$\{\ip3, \ii34 \}$               &                                  & $\{\ip3, \ii14 \}$               \\
			$\boldsymbol{\{\ip3, \ii24 \}}$  &                                  & $\boldsymbol{\{\ip3, \ii24\}}$   \\
			$\boldsymbol{\{\ip3, \ii04 \}}$  &                                  & $\boldsymbol{\{\ip3, \ii44 \}}$  \\
			$\boldsymbol{\{\ip3, \ii44 \}}$  &                                  & $\{\ip3, \ii34 \}$               \\
			$\boldsymbol{\{\ip4, \ii34\}}$   &                                  & $\boldsymbol{\{\ip3, \ii04 \}}$  \\
			$\boldsymbol{\{\ip4, \ii24 \}}$  &                                  & $\boldsymbol{\{\ip4, \ii14 \}}$  \\
			$\boldsymbol{\{\ip4,  \ii04 \}}$ &                                  & $\boldsymbol{\{\ip4, \ii24\}}$   \\
			$\boldsymbol{\{\ip4,  \ii44 \}}$ &                                  & $\boldsymbol{\{\ip4,  \ii44 \}}$ \\
			                                 &                                  & $\boldsymbol{\{\ip4, \ii34 \}}$  \\
			                                 &                                  & $\boldsymbol{\{\ip4, \ii04 \}}$  \\
		 	\hline
		\end{tabular}
		\caption{The spectrum of the $\mathsf{mcs}(r)$, for each $r\in\allm^+$. - Class: $\Den$.}\label{tab:mixedden}
	\end{center}
\end{table}

\begin{lemma}\label{lem:tab:mixedden}
Tab.~\ref{tab:mixedden} is correct.
\end{lemma}

\begin{proof}
As for the $\{\ip 0\}$-completeness, only one new definition is necessary:

\medskip

\begin{small}
\[
\begin{array}{llll}
x_i \ip 0 y_p & \leftrightarrow & \exists k_p(\forall z_p(x_i\ip 2 z_p\rightarrow k_p<z_p)\wedge y_p<k_p) & \{\ip 2,<\}
\end{array}
\]
\end{small}

\medskip

\noindent Let us prove that the set \underline{$\{\ip 2,<\}$} is $\ip 0$-complete. If $\mathcal F\models\varphi(x_i,y_p)$, then there exists a point $k_p$ smaller than every point {\em contained} in $x_i$ (there are infinitely many such points because we are in a dense structure); and $y_p$ is smaller than $k_p$, so it must be {\em before} $x_i$. If, on the other hand $x_i=[a,b]$ and $y_p=c<a$, we take $k_p=a$ to satisfy all requirements. The $\ip 0$-completeness of the remaining new sets is now a consequence of the $\ip 0$-completeness and previously established results. We will give the details of the deductions chains involved. For \underline{$\{\ip 4, \ii 14 \}$}, note that $\{\ip 4\} \rightarrow_{\Den} \: <$ (Table \ref{tab:EqPEqILtDen}) and that $\{\ip 4, \ii 14, < \} \rightarrow_{\Lin} \ip 0$ (Table \ref{tab:alllin1}).
%
%For \underline{$\{\ip 1, \ii 24 \}$}, note that $\{\ii 24\} \rightarrow_{\Den} \: =_i$ (Table \ref{tab:EqPEqILtDen}) and that $\{\ip 1, \ii 24, =_i \} \rightarrow_{\Lin} \ip 0$ (Table \ref{tab:alllin1}). For \underline{$\{\ip 1, \ii 04 \}$}, note that $\{\ii 04\} \rightarrow_{\Den} \: =_i$ (Table \ref{tab:EqPEqILtDen}) and that $\{\ip 1, \ii 04, =_i \} \rightarrow_{\Lin} \ip 0$ (Table \ref{tab:alllin1}). For \underline{$\{\ip 1, \ii 44 \}$}, note that $\{\ii 44\} \rightarrow_{\Den} \: =_i$ (Table \ref{tab:EqPEqILtDen}) and that $\{\ip 1, \ii 44, =_i \} \rightarrow_{\Lin} \ip 0$ (Table \ref{tab:alllin1}).
%
For the cases \underline{$\{\ip 1, \ii 24 \}$}, \underline{$\{\ip 1, \ii 04 \}$} and \underline{$\{\ip 1, \ii 44 \}$}, we note that $\{\ii 24\}$, $\{\ii 04\}$ and $\{\ii 44\}$ all define $=_i$ over $\Den$ (Table \ref{tab:EqPEqILtDen}), and then that  $\{\ip 1, \ii 24, =_i \}$, $\{\ip 1, \ii 04, =_i \}$ and $\{\ip 1, \ii 44, =_i \}$ all define $\ip 0$ over $\Lin$ (Table \ref{tab:alllin1}). The cases for \underline{$\{\ip 3, \ii 24 \}$}, \underline{$\{\ip 3, \ii 04 \}$} and \underline{$\{\ip 3, \ii 44 \}$} follow similarly, using that  $\{\ip 3, \ii 24, =_i \}$, $\{\ip 3, \ii 04, =_i \}$ and $\{\ip 3, \ii 44, =_i \}$ all define $\ip 0$ over $\Lin$ (Table \ref{tab:alllin1}). Considering the $\ip 0$-completeness of \underline{$\{\ip 4, \ii 34 \}$}, we again use the fact that $\{\ip 4\} \rightarrow_{\Den} \: <$ (Table \ref{tab:EqPEqILtDen}) and that $\{\ip 4, \ii 34, < \} \rightarrow_{\Lin} \ip 0$ (Table \ref{tab:alllin1}). In the case of \underline{$\{\ip 4, \ii 24 \}$} the chain of deductions goes as follows: $\{\ip 4\} \rightarrow_{\Den} \: <$ (Table \ref{tab:EqPEqILtDen}), $\{\ii 24\} \rightarrow_{\Den} \: \ii 14$ (Table \ref{tab:mcsIntIntDen}) and hence, by Lemma \ref{lem:symm},  $\{\ii 24\} \rightarrow_{\Den} \: \ii 03$, and lastly  $\{\ip4, \ii 24, \ii 03, < \} \rightarrow_{\Lin} \ip 0$ (Table \ref{tab:alllin1}). In the case of \underline{$\{\ip 4, \ii 04 \}$}, recall that $\{\ip 4, \ii 04 \} \rightarrow_{\Den} \ii 34$ (Lemma \ref{lem:anticipate:mixed:dense}), and then appeal to the just proved $\ip 0$-completeness of $\{\ip 4, \ii 34 \}$. In the case of \underline{$\{\ip 4, \ii 44 \}$} we use the fact that $\{ \ii 44 \} \rightarrow_{\Den} \ii 34$ and again appeal to the case for $\{\ip 4, \ii 34 \}$. As for proving that \underline{$\{\ip 0\}$} is $\ip 1$-complete, observe that by the previous lemma this set defines $<$, and consider the following simple definition:

\medskip

\begin{small}
\[
\begin{array}{llll}
x_i \ip 1 y_p & \leftrightarrow  & \forall z_p(x_i\ip 0 z_p\leftrightarrow z_p<y_p) & \{\ip 0\}\\
\end{array}
\]
\end{small}

\medskip

\noindent The left to right direction of the above definition is immediate. For the right to left direction, suppose, on the
contrary, that the formula holds and it is not the case that $x_i\ip 1 y_p$. Then, assume that $x_i=[a,b]$: if $y_p<a$, then, since the underlying linear order is dense, we could find $c$ such that $y_p<c<a$, and contradict the right-side by instantiating $z_p$ with $c$, and if $a<y_p$, then since the underlying linear order is dense, we could find $c$ such that $a<c<y_p$, and, again, contradict the right-side by instantiating $z_p$ with $c$. Now, the remaining new definitions for the relation $\ip 1$ all follow straightforwardly by noting that these sets all define $\ip 0$.  Finally, we prove that $\{\ip 1,\ii04\}$ is $\ip 2$-complete by means of the following definability equation:

\medskip

\begin{small}
\[
\begin{array}{llll}
x_i \ip 2 y_p & \leftrightarrow  & \exists z_i(z_i\ip 1 y_p\wedge z_i\ii 04 x_i). & \{\ip 1,\ii 04\}
\end{array}
\]
\end{small}

\medskip

\noindent The $\ip 2$-completeness of  \underline{$\{\ip 1,\ii 04\}$} is extremely simple.
Assuming $\mathcal{F} \models \varphi([a,b],c)$, there must be an interval $z_i=[c,d]$ {\em contained} in $[a,b]$,
which implies $a<c<b$, as we wanted. Conversely, if $a<c<b$, since the structure is dense, we can find $d$ such that $a<c<d<b$, and $z_i=[c,d]$ is a witness for $\varphi$. For the $\ip 2$-completeness of  \underline{$\{\ip 0,\ii 24\}$},  \underline{$\{\ip 0,\ii 44\}$} and  \underline{$\{\ip 0,\ii 34\}$} we note that facts that $\{\ip 0\} \rightarrow_{\Den} \ip 1$, that $\{\ii 24\}$, $\{\ii 44\}$ and $\{\ii 34\}$ each define $\ii 04$ over $\Den$, and the just proved  $\ip 2$-completeness of $\{\ip 1,\ii 04\}$. For \underline{$\{\ip 0,\ii 03\}$} we note that $\{\ip 0\} \rightarrow_{\Den} \: <$ and that $\{\ip 0,\ii 03, < \} \rightarrow_{\Lin} \ip 2$. For \underline{$\{\ip 1,\ii 24\}$} and  \underline{$\{\ip 1,\ii 44\}$} we again use the fact that $\{\ii 24\}$ and $\{\ii 44\}$ each define $\ii 04$ over $\Den$ and that $\{\ip 1,\ii 04\} \rightarrow_{\Den} \ii 02$ as proved above. For  \underline{$\{\ip 3,\ii 24\}$} and  \underline{$\{\ip 3,\ii 44\}$} we use the facts that $\{\ii 24\}$ and $\{\ii 44\}$ each defines $\ii 14$ over $\Den$ and that $\{\ip 3, \ii 14 \} \rightarrow_{\Lin} \ip 2$. In the case of \underline{$\{\ip 3,\ii 04\}$} we use the facts that $\{\ip 3,\ii 04\} \rightarrow_{\Den} \ii 34$ (Lemma \ref{lem:anticipate:mixed:dense}) and that $\{\ip 3,\ii 34\} \rightarrow_{\Lin} \ip 2$. For \underline{$\{\ip 4,\ii 34\}$} we use the previously proven facts that $\{\ip 4,\ii 34\} \rightarrow_{\Den} \ip 1$ and $\{\ip 1,\ii 34\} \rightarrow_{\Lin} \ip 2$. From this the cases for \underline{$\{\ip 4,\ii 24\}$}, \underline{$\{\ip 4,\ii 44\}$} and \underline{$\{\ip 4,\ii 04\}$} now respectively follow since $\{\ii 24\}  \rightarrow_{\Den} \ii 34$, $\{\ii 44 \} \rightarrow_{\Den} \ii 34$ and $\{\ip 4,\ii 04\} \rightarrow_{\Den} \ii 34$ (Lemma \ref{lem:anticipate:mixed:dense}). Lastly the $\ip 2$-completeness of \underline{$\{\ip 4,\ii 14\}$} follows by noting that $\{\ip 4 \} \rightarrow_{\Den} \: <$ and that $\{\ip 4,\ii 14, < \} \rightarrow_{\Lin} \ip 2$.
\end{proof}

\subsection{Definability for Relations in $\alli^+\setminus\{=_i\}$ in $\Den$}\label{sec:All:i:den}

\begin{table}[t]
\begin{center}
\small
\begin{tabular}{|p{0.17\textwidth}|p{0.17\textwidth}|p{0.17\textwidth}|p{0.17\textwidth}|p{0.17\textwidth}|}
	\hline
	$\ii 14$                        & $\ii 34$                        & $\ii 24$                        & $\ii 04$                  & $\ii 44$                        \\ \hline
	$\{\ip0, \ip2 \}$               & $\{\ip0, \ip2 \}$               & $\{\ip0, \ip2 \}$               & $\{\ip0, \ip3 \}$         & $\{\ip0, \ip2 \}$               \\
	$\{\ip0, \ip3 \}$               & $\{\ip0, \ip3 \}$               & $\{\ip0, \ip3 \}$               & $\{\ip0, \ip4 \}$         & $\{\ip0, \ip3 \}$               \\
	$\{\ip0, \ip4 \}$               & $\{\ip0, \ip4 \}$               & $\{\ip0, \ip4 \}$               & $\{\ip0, \ii03 \}$        & $\{\ip0, \ip4 \}$               \\
	$\{\ip0, \ii03 \}$              & $\{\ip0, \ii03 \}$              & $\{\ip0, \ii03 \}$              & $\{\ip1, \ip3 \}$         & $\{\ip0, \ii03 \}$              \\
	$\{\ip0, \ii04 \}$              & $\boldsymbol{\{\ip0, \ii04 \}}$ & $\boldsymbol{\{\ip0, \ii04 \}}$ & $\{\ip1, \ip4 \}$         & $\boldsymbol{\{\ip0, \ii04 \}}$ \\
	$\{\ip1, \ip2 \}$               & $\{\ip1, \ip2 \}$               & $\{\ip1, \ip2 \}$               & $\{\ip1, \ii03 \}$        & $\{\ip1, \ip2 \}$               \\
	$\{\ip1, \ip3 \}$               & $\{\ip1, \ip3 \}$               & $\{\ip1, \ip3 \}$               & $\boldsymbol{\{\ip2 \}}$  & $\{\ip1, \ip3 \}$               \\
	$\{\ip1, \ip4 \}$               & $\{\ip1, \ip4 \}$               & $\{\ip1, \ip4 \}$               & $\{\ip3, \ii14 \}$        & $\{\ip1, \ip4 \}$               \\
	$\{\ip1, \ii03 \}$              & $\{\ip1, \ii03 \}$              & $\{\ip1, \ii03 \}$              & $\{\ip4, \ii14 \}$        & $\{\ip1, \ii03 \}$              \\
	$\boldsymbol{\{\ip1, \ii04 \}}$ & $\boldsymbol{\{\ip1, \ii04 \}}$ & $\{\ip1, \ii04 \}$              & $\{\ii14, \ii03 \}$       & $\boldsymbol{\{\ip1, \ii04 \}}$ \\
	$\{\ip2, \ip3 \}$               & $\{\ip2, \ip3 \}$               & $\{\ip2, \ip3 \}$               & $\boldsymbol{\{\ii24 \}}$ & $\{\ip2, \ip3 \}$               \\
	$\{\ip2, \ip4 \}$               & $\{\ip2, \ip4 \}$               & $\{\ip2, \ip4 \}$               & $\{\ii34 \}$              & $\{\ip2, \ip4 \}$               \\
	$\{\ip2, \ii03 \}$              & $\{\ip2, \ii14 \}$              & $\{\ip2, \ii14 \}$              & $\boldsymbol{\{\ii44 \}}$ & $\{\ip2, \ii14 \}$              \\
	$\boldsymbol{\{\ip2, < \}}$     & $\{\ip2, \ii03 \}$              & $\{\ip2, \ii03 \}$              &                           & $\{\ip2, \ii03 \}$              \\
	$\boldsymbol{\{\ip3, \ii04 \}}$ & $\boldsymbol{\{\ip2, < \}}$     & $\{\ip2, < \}$                  &                           & $\boldsymbol{\{\ip2, < \}}$     \\
	$\boldsymbol{\{\ip4, \ii04 \}}$ & $\{\ip3, \ii14 \}$              & $\{\ip3, \ii14 \}$              &                           & $\{\ip3, \ii14 \}$              \\
	$\{\ii24 \}$                    & $\boldsymbol{\{\ip3, \ii04 \}}$ & $\boldsymbol{\{\ip3, \ii04 \}}$ &                           & $\boldsymbol{\{\ip3, \ii04 \}}$ \\
	$\{\ii03, \ii04 \}$             & $\{\ip4, \ii14 \}$              & $\{\ip4, \ii14 \}$              &                           & $\{\ip4, \ii14 \}$              \\
	$\{\ii34 \}$                    & $\boldsymbol{\{\ip4, \ii04 \}}$ & $\boldsymbol{\{\ip4, \ii04 \}}$ &                           & $\boldsymbol{\{\ip4, \ii04 \}}$ \\
	$\boldsymbol{\{\ii44 \}}$       & $\{\ii14, \ii03 \}$             & $\{\ii14, \ii03 \}$             &                           & $\{\ii14, \ii03 \}$             \\
	                                & $\{\ii14, \ii04 \}$             & $\{\ii14, \ii04 \}$             &                           & $\{\ii14, \ii04 \}$             \\
	                                & $\boldsymbol{\{\ii24 \}}$       & $\{\ii03, \ii04 \}$             &                           & $\boldsymbol{\{\ii24 \}}$       \\
	                                & $\{\ii03, \ii04 \}$             & $\{\ii34 \}$                    &                           & $\{\ii03, \ii04 \}$             \\
	                                & $\boldsymbol{\{\ii44 \}}$       & $\boldsymbol{\{\ii44 \}}$       &                           & $\{\ii34 \}$                    \\
	                                &                                 &                                 &                           & $\boldsymbol{\{\ii44 \}}$       \\ \hline
\end{tabular}
\end{center}\caption{The spectrum of the $\mathsf{mcs}(r)$, for each $r\in\alli^+\setminus\{=_i\}$. - Class: $\Den$.}\label{tab:iden}
\end{table}

To conclude this analysis of dense linear orders, we prove that Tab.~\ref{tab:iden} is correct.

\medskip

\begin{lemma}\label{lem:tab:iden}
Tab.~\ref{tab:iden} is correct.
\end{lemma}

\begin{proof}
All new definitions involving only interval-interval relations have already been treated in Lemma \ref{lem:tab:mcsIntIntDen}. All new $\ii 14$-completeness results follow from previously proven facts, and we will give the deduction chains. For \underline{$\{\ip 2, < \}$}, we have that $\{\ip 2, < \}$ defines $\ip 0$ in the dense case and that $\{\ip 0, \ip 2\}$ defines $\ii 14$. The cases for \underline{$\{\ip 3, \ii 04 \}$} and \underline{$\{\ip 4, \ii 04 \}$} follow from Lemma \ref{lem:anticipate:mixed:dense} and the fact that $\{\ii 34 \}$ define $\ii 14$. The new $\ii 34$-completeness results are based on the following arguments. For \underline{$\{\ip 0, \ii 04 \}$}, we have that $\{\ip 0, \ii 04 \}$ defines $\ip 2$ in the dense case and that $\{\ip 0, \ip 2 \}$ defines $\ii 34$. For \underline{$\{\ip 1, \ii 04 \}$}, we know that $\{\ip 1, \ii 04 \}$ defines $\ip 2$ in the dense case and that $\{\ip 1, \ip 2 \}$ defines $\ii 34$. For \underline{$\{\ip 2, < \}$}, we have that $\{\ip 2, < \}$ defines $\ip 0$ in the dense case and that  $\{\ip 0, \ip 2\}$ is $\ii 34$-complete.  The cases for \underline{$\{\ip 3, \ii 04 \}$} and \underline{$\{\ip 4, \ii 04 \}$} were proven in Lemma \ref{lem:anticipate:mixed:dense}. The $\ii 24$-completeness of the sets \underline{$\{\ip 0, \ii 04 \}$}, \underline{$\{\ip 3, \ii 04 \}$}  and \underline{$\{\ip 4, \ii 04 \}$} follow from their already proven $\ii 34$-completeness over $\Den$ and the $\ii 24$-completeness of $\ii 34$ over $\Lin$. For the $\ii 04$-completeness of \underline{$\{\ip 2\}$}, assuming density, recall that $\{\ip 2\}$ is $=_i$-complete in the dense case, and consider the following definition:

\medskip

\begin{small}
	\[
	\begin{array}{lllll}
	x_i\ii 04 y_i   &\leftrightarrow   &&\forall z_p(x_i\ip 2 z_p\rightarrow y_i\ip 2 z_p)\wedge\exists z_p(y_i\ip 2 z_p\wedge \neg(x_i\ip 2 z_p)) & \{\ip 2\}\\
	& &\wedge &\exists z_i \exists z'_i \left[%
		\begin{array}{llll}
		&\neg z_i =_i z'_i\\
		\wedge &\forall u_p [(z_i \ip2 u_p \vee z'_i \ip2 u_p) \rightarrow (y_i \ip2 u_p \wedge \neg x_i \ip2 u_p)]\\
		\wedge &\forall v_i \left[%
			\begin{array}{lll}
			&(\forall u_p(z_i \ip2 u_p \vee z'_i \ip2 u_p \rightarrow v_i \ip2 u_p))\\
			\rightarrow &\exists t_p (x_i \ip2 t_p \wedge v_i \ip2 t_p)
			\end{array}\right]
		\end{array}
		\right]
	\end{array}
	\]
\end{small}

\medskip

\noindent First suppose that $\mathcal{F} \models \varphi([a,b],$ $[c,d])$. Then $[c,d]$ must {\em contain} every point {\em contained} by $[a,b]$, and there must be some point {\em contained} in $[c,d]$ which is not {\em contained} in $[a,b]$. So we have $c\le a$, $b\le d$ and $[a,b]\neq [c,d]$. (Note that here we are using the fact, which follows from the density of the order, that $[a,b]$ contain at least one point. If this were not the case, the first conjunct of $\phi$ would be vacuously satisfied.)  We need to show that $a\neq c$ and $d\neq b$. Suppose, to the contrary, that $a = c$. Then we must have that $b < d$. Let $\psi$ denote the third main conjunct of $\varphi$ and let $z_i = [e,f]$ and $z'_i = [g,h]$ be two intervals satisfying $\psi$. Then $[e,f] \neq [g,h]$ by the first conjunct of $\psi$, while $b \leq e < f \leq d$ and $b \leq g < h \leq d$. However now, taking $v_i = [b,d]$ falsifies the third conjunct of $\psi$, as $[b,d]$ contains all points contained in either $[e,f]$ or $[g,h]$ while not containing any point contained in $x_i = [a,b]$. This is a contradiction. Conversely, if $[a,b]\ii 04 [c,d]$, the the first conjunct of $\varphi$ holds trivially, while any point $e$ such that $c<e<a$ or $b<e<d$ (which exists because the structure is dense) witnesses the $z_p$ of the second conjunct. Taking $z_i = [c,a]$ and $z'_i = [b,d]$ witnesses the third conjunct. All that remains to be proved is the $\ii 44$-completeness of the sets \underline{$\{\ip 0, \ii 04\}$}, \underline{$\{\ip 1, \ii 04\}$}, \underline{$\{\ip 2, < \}$}, \underline{$\{\ip 3, \ii 04\}$} and \underline{$\{\ip 4, \ii 04\}$}. But all of these sets have already been shown to define $\ii 34$ over $\Den$, and we know that $\{\ii 34\}$ defines $\ii 44$ over $\Lin$.
\end{proof}

\section{Incompleteness Results in The Class $\Den$}\label{sec:den-incomp}

We can now turn our attention to the maximal incomplete sets for relations in $\allr^+$. Notice that for some $r\in\allr^+$ some $r$-incomplete
set in the class $\Lin$ is also maximally $r$-incomplete in the class $\Den$; nevertheless, all proofs in Part I must be revisited
here, as they were designed to work not only in $\Lin$ but also in $\Dis$, and therefore are not valid in $\Den$.

\begin{table}[t]
\small
\begin{center}
\begin{tabular}{p{0.34\textwidth}|p{0.022\textwidth}|p{0.022\textwidth}|p{0.022\textwidth}|p{0.022\textwidth}|p{0.022\textwidth}|p{0.022\textwidth}|p{0.022\textwidth}|p{0.022\textwidth}|p{0.022\textwidth}|p{0.022\textwidth}|p{0.022\textwidth}|p{0.022\textwidth}|p{0.022\textwidth}|p{0.022\textwidth}|p{0.022\textwidth}|}
{\em Proved}                            & $=_p$             & $=_i$             & $<$              & $\ip0$           & $\ip1$           & $\ip2$          & $\ip3$           & $\ip4$           & $\ii34$          & $\ii14$          & $\ii03$          & $\ii24$          & $\ii04$          & $\ii44$           \\
 \hline
$\{\ip 2\}\cup\alli^{+}$                & $\bullet$         &                   &                  &                  &                  &                 &                  &                  &                  &                  &                  &                  &                  &                   \\
 \hline
$\{=_p, <, \ip 0,\ip 1\}$               &                   & $\bullet$         &                  &                  &                  &                 &                  &                  &                  &                  &                  &                  &                  &                   \\
\hline
$\{=_p, =_i, \ip1, \ii 14\}$            &                   &                   & $\bullet$        & $\bullet$        &                  &                 &                  &                  &                  &                  &                  &                  &                  &                   \\
\hline
$\{=_p, \ip 2\}\cup\alli^+$             &                   &                   & $\bullet$        & $\bullet$        & $\bullet$        &                 &  $\bullet$                &  $\bullet$                &                  &                  &                  &                  &                  &                   \\
\hline
$\{=_p, =_i, <, \ip 3, \ip 4, \ii 03\}$ &                   &                   &                  & $\bullet$        & $\bullet$        & $\bullet$       &                  &                  & $\bullet$        & $\bullet$        &                  & $\bullet$        & $\bullet$        &  $\bullet$        \\
\hline
$\{=_p, <\}\cup\alli^+$                 &                   &                   &                  & $\bullet$        & $\bullet$        & $\bullet$       &     $\bullet$             & $\bullet$                 &                  &                  &                  &                  &                  &                   \\
\hline
$\{=_p, =_i, <, \ip 0, \ip 1\}$         &                   &                   &                  &                  &                  &                 &                  &                  &                  & $\bullet$        &                  &                  &                  &                   \\
\hline
$\{=_p, =_i, \ip 2, \ii 04\}$           &                   &                   &                  &                  &                  &                 &                  &                  & $\bullet$        & $\bullet$        &  $\bullet$       & $\bullet$        &                  & $\bullet$         \\
\hline
$\{=_p, =_i, <,   \ii 04\}$             &                   &                   &                  &                  &                  &                 &                  &                  & $\bullet$        & $\bullet$        &  $\bullet$       & $\bullet$        &                  & $\bullet$         \\
\hline
\\
{\em Symmetric}                          & $=_p$             & $=_i$              & $<$              & $\ip0$           & $\ip1$           & $\ip2$              & $\ip3$           & $\ip4$           & $\ii34$          & $\ii14$          & $\ii03$          & $\ii24$          & $\ii04$          & $\ii44$           \\
\hline
$\{=_p, <, \ip 3,\ip 4\}$                &                   &$\bullet$           &                  &                  &                  &                     &                  &                  &                  &                  &                  &                  &                  &                   \\
\hline
$\{=_p, =_i, \ip 3, \ii 03\}$            &                   &                    &$\bullet$         &                  &                  &                     &                  & $\bullet$        &                  &                  &                  &                  &                  &                   \\
\hline
$\{=_p, =_i, <,\ip 0, \ip 1, \ii 14\}$   &                   &                    &                  &                  &                  & $\bullet$           & $\bullet$        & $\bullet$        & $\bullet$        &                  & $\bullet$        & $\bullet$        & $\bullet$        & $\bullet$         \\
\hline
$\{=_p, =_i, <, \ip 3, \ip 4\}$          &                   &                    &                  &                  &                  &                     &                  &                  &                  &                  & $\bullet$        &                  &                  &                   \\
\hline

\end{tabular}                                                                                                                                                                                                                                         \caption{$\mathsf{MIS}(r)$, for each $r\in\allr^+$; upper part: sets for which we give an explicit construction; lower part: symmetric ones. - Class: $\Den$.}\label{tab:MISRden}
\end{center}
\end{table}

\begin{lemma}\label{lem:tab:MISRden}
Tab.~\ref{tab:MISRden} is correct.
\end{lemma}

\begin{proof}
Let $S$ be \underline{$\{\ip 2\}\cup\alli^{+}$}: proving that it is $=_p$-incomplete is almost immediate. Indeed, it suffices to take $\mathbb D$ and $\mathbb D'$ both equal to the subset of $\mathbb Q$ of all points between $0$ and $1$, $\zeta=(\zeta_p,\zeta_i)$, where $\zeta_i=Id_i$ (the identical relation on intervals), $\zeta_p=\{(0,1')\}$ plus the identical relation on points to have a surjective truth-preserving relation that breaks $=_p$. Proving that \underline{$\{=_p,<,\ip 0,\ip 1\}$} is $=_i$-incomplete is equally easy: it suffices to take $\mathbb D=\mathbb D'=\mathbb Q$, $\zeta_p=Id_p$, and $\zeta_i=Id_i$ plus $\zeta_i([0,2],[0',1'])$. Assume, now, $S$ to be \underline{$\{=_p, =_i, \ip 1, \ii 14\}$}; we need to prove its $<,\ip 0$-incompleteness in the dense case. Let $\mathbb D=\mathbb D'=\mathbb Q$, and define $\zeta=\zeta_p\cup\zeta_i$ as follows: $(a,-a')\in\zeta_p$ for every $a\in\mathbb Q$, and $([a,b],[-a',-a'+|b'-a'|])\in\zeta_i$ for every $[a,b]\in\mathbb I(\mathbb Q)$, so that the length of every interval is preserved while their beginning point are reflected over 0; as this relation breaks both $<$ and $\ip 0$, the latter cannot be expressed in this language. Let now $S$ be $\{=_p, \ip2 \}\cup\alli^{+}$: we can prove that it is $<,\ip 0,\ip 1,\ip3,\ip 4$-incomplete. To this end, it suffices to take, once again, $\mathbb D$ and $\mathbb D'$ both equal to the subset of $\mathbb Q$ of all points between $0$ and $1$, $\zeta=(\zeta_p,\zeta_i)$, where $\zeta_i=Id_i$ (the identical relation on intervals), $\zeta_p=\{(0,1'),(1,0')\}$ plus the identical relation on every other point to have a surjective truth-preserving relation that breaks the relations under analysis. As for S=\underline{$\{=_p, =_i, <, \ip 3, \ip 4, \ii 03\}$}, we can prove its $\ip 0,\ip 1,\ip 2,i$-incompleteness, where $i\in\alli^{+}\setminus\{=_i,\linebreak\ii 03\}$ by defining two $\mathbb Q$-based structures and a relation between them defined as the identity between points and as $\zeta_i([a,b])=[a'-|b'-a'|,b']$, obtaining (as we did for the same set of relations on $\Lin$, using, in that case, a pseudo-discrete structure) a relation that maps every interval to the interval with the same {\em ending} point but twice the length. In this way, all relations in $S$ are respected. The $m$-incompleteness of \underline{$\{=_p,<\}\cup\alli^{+}$} for each $m\in\mathfrak M^+$ can be proved by taking again $\mathbb D=\mathbb D'=\mathbb Q$, $\zeta=(\zeta_p,\zeta_i)$, where $\zeta_i=Id_i$ and $\zeta_p(a)=a'+1$ for every $a\in\mathbb Q$, which clearly respects all interval-interval relations, and both equality and relative ordering between points, but breaks every relation between points and intervals. When $S$ is \underline{$\{=_p, =_i, <, \ip 0,\ip 1\}$}, we have to prove that it is $\ii 14$-incomplete. Consider two structures based on $\mathbb Q$, and let $\zeta=(\zeta_p,\zeta_i)$ be defined as $\zeta_p=Id_p$, and $\zeta_i=Id_i$ except for the interval $[-1,0]$, which is mapped to $[-1,1]$.  When $S$ is \underline{$\{=_p, =_i, \ip 2, \ii 04\}$}, we have to prove that it is $r$-incomplete in the dense case for $r\in\alli^{+}\setminus\{\ii 04,=_i\}$. Consider two structures based on $\mathbb Q$, and let $\zeta=(\zeta_p,\zeta_i)$ be defined as $\zeta_p(a)=-a'$ for every point and $\zeta_i([a,b])=[-b',-a']$ for every interval. Clearly, {\em containment} is respected for both sorts; nevertheless, all other interval-interval relations are broken. Finally, when $S$ is \underline{$\{=_p, =_i, <, \ii 04\}$}, we have to prove that it is $r$-incomplete for $r\in\alli^{+}\setminus\{=_i,\ii 04\}$. Consider two structures based on $\mathbb Q$, and let $\zeta=(\zeta_p,\zeta_i)$ be defined as $\zeta_p=Id_p$, and $\zeta_i([a,b])=[-b,-a]$; again, we respect {\em containment}  between intervals, and the relative ordering between points is respected as well (since points are not affected by the construction), and we break every other interval-interval relation.
\end{proof}

\section{Completeness Results in the Class $\Unb$}\label{sec:unb}

The ability of fragments of our language to define relations when the underlying linear order is unbounded (but not necessarily discrete or dense) differs from the dense/discrete cases only slightly. Following the same schema, we now focus on the definability part, again, pointing out the differences with the linear case.

\subsection{Definability in  $\alli^+$ and in $\allm^+$.}

We begin our study of definability over the Class $\Unb$ by considering minimal definability of $\alli^+$ relations by sets of $\alli^+$ relations, and then of $\allm^+$ relations by sets of $\allm^+$ relations.

\begin{table}[t]
	\begin{center}
		\small
		\begin{tabular}{|p{0.17\textwidth}|p{0.17\textwidth}|p{0.17\textwidth}|p{0.17\textwidth}|p{0.17\textwidth}|}
			\hline
			$\ii 14$					&$\ii 34$						& $\ii 24$						&$\ii 04$		&$\ii44$\\
			\hline
			$\{\ii24,\ii03\}$ 			&$\{\ii14,\ii24\}$				&$\{\ii14,\ii03\}$				&$\{\ii14,\ii24\}$				&$\{\ii14,\ii24\}$		\\
			$\{\ii03,\ii04\}$ 			&$\{\ii14,\ii03\}$				&$\{\ii14,\ii04\}$				&$\{\ii14,\ii03\}$				&$\{\ii14,\ii03\}$		\\
			$\{\ii34\}$ 				&$\{\ii14,\ii04\}$				&$\{\ii03,\ii04\}$				&$\{\ii24,\ii03\}$				&$\{\ii14,\ii04\}$		\\
			$\boldsymbol{\{\ii44\}}$ 	&$\{\ii24,\ii03\}$				&$\{\ii34\}$					&$\{\ii34\}$				&$\{\ii24,\ii03\}$		\\
			 							&$\boldsymbol{\{\ii24,\ii04\}}$	&$\boldsymbol{\{\ii44\}}$		&$\boldsymbol{\{\ii44\}}$				&$\boldsymbol{\{\ii24,\ii04\}}$		\\
			 							&$\{\ii03,\ii04\}$				&				&				&$\{\ii03,\ii04\}$		\\
			 							&$\boldsymbol{\{\ii44\}}$ 		&				&				&$\{\ii34\}$		\\
			
			\hline
		\end{tabular}
	\end{center}\caption{The spectrum of the $\mathsf{mcs}(r)$, for each $r\in\alli^+\setminus\{=_i\}$. - Class: $\Unb$.}\label{tab:Int:Int:unb}
\end{table}

\begin{lemma}\label{lem:Pure:Int:Inte:Und}
	Tab.~\ref{tab:Int:Int:unb} is correct.
\end{lemma}

\begin{proof}
To prove that \underline{$\{\ii 44\}$} is $\ii 14$-complete we ue the same definition as was used in the dense case in Lemma \ref{lem:tab:mcsIntIntDen}:
	
	\medskip
	
	\begin{small}
		\[
		\begin{array}{llll}
		x_i \ii 14 y_i & \leftrightarrow & \forall z_i(z_i\ii 44 x_i\leftrightarrow z_i\ii 44 y_i)\wedge\exists z_i(x_i\ii 44 z_i \wedge\neg(y_i\ii 44 z_i)). & \{\ii 44\}
		\end{array}
		\]
	\end{small}

\medskip

\noindent The argument for the correctness of this definition is the same as in the dense case except that, when arguing that $[a,b] \ii14 [c,d]$ implies $\phi([a,b],[c,d])$, we appeal to unboundedness rather than density to justify the existence of the interval required by the second conjunct of the definition. As before, the $\ii34$-completeness of \underline{$\{\ii 44\}$} now follows via $\ii34$-completeness of $\{\ii14, \ii44 \}$ over all linear orders. Now, since $\ii34$ defines all interval-interval relations over all linear orders, it follows that $\ii 44$ is also complete with respect all other interval-interval relations over unbounded orders. To compete the proof, it suffices to show the $\ii 34$-completeness of \underline{$\{\ii 24,\ii 04\}$}. We do so by proving, as we did in Lemma 3.1 of Part I of this paper, that this set is able to express the weaker relation $\ii 34\cup\ii 44$:
	
	\medskip
	
	\begin{small}
		\[
		\begin{array}{llll}
		x_i\ii 34\cup\ii 44 y_i  &\leftrightarrow&  \exists z_i,k_i(x_i\ii 04 z_i\wedge y_i\ii 04 k_i\wedge z_i\ii 24 k_i)\wedge & \{\ii04,\ii 24\}\\
		&&\exists z_i(x_i\ii 04 z_i\wedge\neg(y_i\ii 04 z_i))\wedge \exists k_i(y_i\ii 04 k_i\wedge\neg(x_i\ii 04 k_i))\wedge & \\
		&&\neg(x_i\ii 24 y_i)\wedge\neg(y_i\ii 24 x_i)\wedge\neg(x_i\ii 04 y_i)\wedge\neg(y_i\ii 04 x_i).
		\end{array}
		\]
	\end{small}
	
	\medskip
	
\noindent If $\mathcal{F} \models \varphi([a,b],[c,d])$ over an unbounded structure, we can eliminate all possibilities for the relationship between $[a,b]$ and $[c,d]$ other than $\ii 34$ and $\ii 44$. Indeed, first observe that $x_i$ and $y_i$ cannot {\em  overlap} nor {\em contain} each other. Next, if $y_i$ {\em ends} before $x_i$ or at its {\em beginning} point, it would be impossible to place $z_i$ and $k_i$. Finally, if $x_i$ {\em starts} or {\em finishes} $y_i$, or the other way around, we have a contradiction with the second or the third requirement of $\varphi$. This implies that $x_i$ {\em meets} or is {\em before} $y_i$. Conversely, if $[a,b]\ii 34\vee\ii 44[c,d]$, then we can take $z_i=[e,d]$, where $e<a$, and $k_i=[a,f]$, where $d<f$, and the existence of $e$ and $f$ is guaranteed by the assumption of unboundedness.
\end{proof}

\subsection{Definability for the Relations $=_p,=_i,<$ in $\Unb$}

\begin{table}[t]
\small
\begin{center}
\begin{tabular}{|p{0.20\textwidth}|p{0.20\textwidth}|p{0.20\textwidth}|}
\hline
$=_p$                     & $=_i$                    & $<$                     \\
\hline
$\{<\}$                   & $\{\ip 0,\ip 2\}$              &  $\boldsymbol{\{\ip0 \}}$     \\
$\boldsymbol{\{\ip 0\}}$  & $\{\ip 0,\ip 3\}$              &  $\{\ip1, \ip2 \}$         \\
$\boldsymbol{\{\ip 2\}}$  & $\{\ip 0,\ip 4\}$              &  $\{\ip1, \ip3 \}$          \\
$\{\ip 1\}$                & $\boldsymbol{\{\ip 0,\ii24\}}$ &             $\boldsymbol{\{\ip1, \ii24 \}}$ \\
$\{\ip 3\}$                & $\boldsymbol{\{\ip 1,\ii24\}}$ &   $\{\ip1, \ii03 \}$     \\                     
$\boldsymbol{\{\ip 4\}}$   & $\{\ip 1,\ip 2\}$              &    $\{\ip1, \ii34 \}$     \\                    
                          & $\{\ip 1,\ip 3\}$ 			   &          $\boldsymbol{\{\ip1, \ii04 \}}$     \\        
                          & $\{\ip 1,\ip 4\}$	 		   &          $\boldsymbol{\{\ip1, \ii44 \}}$     \\        
                          & $\{\ip 2,\ip 3\}$	 		   &                    $\{\ip2, \ip3 \}$     \\
                          & $\{\ip 2,\ip 4\}$	 		   &                $\boldsymbol{\{\ip2, \ii14 \}}$   \\                          
                          & $\boldsymbol{\{\ip 3,\ii24\}}$ &          $\boldsymbol{\{\ip2, \ii24 \}}$    \\                        
                          & $\boldsymbol{\{\ip 4,\ii24\}}$ &            $\boldsymbol{\{\ip2, \ii03 \}}$      \\                    
                          & $\{\ii 14\}$	 			   &                        $\boldsymbol{\{\ip2, \ii34 \}}$        \\                  
                          & $\{\ii 03\}$                   &             $\boldsymbol{\{\ip 2,\ii44\}}$          \\                 
                          & $\{\ii 34\}$                   &           $\{\ip3, \ii14 \}$         \\                                
                          & $\boldsymbol{\{\ii04\}}$       &          $\boldsymbol{\{\ip3, \ii24 \}}$         \\                    
                          & $\boldsymbol{\{\ii44\}}$       &          $\{\ip3, \ii34 \}$           \\                               
                          &                                &            $\boldsymbol{\{\ip3, \ii04 \}}$         \\                                  
                          &                           	   &           $\boldsymbol{\{\ip3, \ii44 \}}$         \\                    
                          &                          	   &              $\boldsymbol{\{\ip4 \}}$          \\                         
\hline
\end{tabular}
\caption{The spectrum of the $\mathsf{mcs}(r)$, for each $r\in\{=_p,=_i,<\}$. - Class: $\Unb$.}\label{tab:EqPEqILtUnb}
\end{center}
\end{table}

\begin{lemma}
Tab.~\ref{tab:EqPEqILtUnb} is correct.
\end{lemma}

\begin{proof}
Starting with $=_p$, we have now that every mixed relation is $=_p$-complete, as follows:

\medskip

\begin{small}
\[
\begin{array}{llll}
x_p=_p y_p & \leftrightarrow & \forall x_i(x_i~m~x_p\leftrightarrow x_i~m~y_p). & \{m\}, ~m\in\mathfrak M^+
\end{array}
\]
\end{small}

\medskip

\noindent Observe that this definition is the same that we have used in the dense case for $\ip 0$ and $\ip 4$; the difference is that now, because of the unboundedness hypothesis, the argument also works for $\underline{\{\ip 2\}}$, as we now prove. The left to right direction if immediate. We argue contrapositively for the other implication. Suppose that $x_p=a$ and $y_p=b$ are not equal. Without loss of generality we can assume that $x_p<y_p$. Since the underlying domain is unbounded, there must be a point $c<a$, and therefore, the interval $[c,b]$ is such that $[c,b]\ip 2 a$ but it is not the case that $[c,b]\ip 2 b$, falsifying the right-hand side. Notice that this argument does not work on dense structures that are left/right bounded, such as $[0,1]\subset\mathbb Q$. Now, consider the relation $=_i$, and the following definitions:

\medskip

\begin{small}
\[ x_i =_i y_i \leftrightarrow \left\{ \begin{array}{ll}
                                                                                                                                                         \forall z_p (x_i \ip 0 z_p \leftrightarrow y_i \ip 0 z_p) \wedge  \forall u_i (x_i \ii 24 u_i \leftrightarrow y_i \ii 24 u_i)& \{\ip0,\ii 24\}\\
                                  & \\
                                                                                                                                                         \forall z_p (x_i \ip 1 z_p \leftrightarrow y_i \ip 1 z_p) \wedge  \forall u_i (x_i \ii 24 u_i \leftrightarrow y_i \ii 24 u_i) & \{\ip1,\ii 24\}\\
                                   & \\                                                                                                                      
                                   \forall z_i(x_i\ii 04 z_i\leftrightarrow y_i\ii 04 z_i)\wedge \forall z_i(z_i\ii 04 x_i\leftrightarrow z_i\ii 04 y_i) & \{\ii 04\} \\
                                   & \\
                                   \forall z_i(x_i\ii 44 z_i\leftrightarrow y_i\ii 44 z_i)\wedge \forall z_i(z_i\ii 44 x_i\leftrightarrow z_i\ii 44 y_i). & \{\ii 44\} \\
   \end{array} \right. \]
\end{small}

\medskip

\noindent First take the set $\underline{\{\ip0,\ii 24\}}$. If $x_i =_i y_i$ we immediately have $\varphi(x_i, y_i)$. Conversely, suppose $x_i = [a,b] \neq_i [c,d] = y_i$, so either $a \neq c$ or $b \neq d$. In the first case the first conjunct of $\varphi$ is falsified. Assume therefore that $a = c$ but that $b \neq d$ and, w.l.o.g., that $b < d$. Since the order is unbounded, there is a point $e > d$. Then $[c,d] \ii 24 [b,e]$ but $\neg ([a,b] \ii 24 [b,e])$, falsifying the second conjunct. The argument in the case of $\underline{\{\ip1,\ii 24\}}$ is virtually identical. The $=_i$-completeness of $\underline{\{\ip4,\ii 24\}}$ and $\underline{\{\ip3,\ii 24\}}$ follows by symmetry (Lemma \ref{lem:symm}).  Consider the set \underline{$\{\ii 04\}$}. If $x_i =_i y_i$ we immediately have $\varphi(x_i, y_i)$. Conversely, suppose that $[a,b] \neq_i [c,d]$, so either $a \neq c$ or $b \neq d$. Consider, w.l.o.g., the case $a \neq c$, specifically $a < c$.  We then choose a point $e>\max\{b,d\}$, which does exist because the underlying domain is unbounded. Now, we have that $[c,d]\ii 04 [a,e]$, but it is not the case that $[a,b]\ii 04 [a,e]$, falsifying the right-hand side. The other case is treated with a similar argument. Let us now focus on $<$. Five new definitions are needed:

\medskip

\begin{small}
\[ x_p < y_p \leftrightarrow \left\{ \begin{array}{ll}
                                   \forall x_i(x_i\ip 0 y_p\rightarrow x_i\ip 0 x_p)\wedge\exists x_i(x_i\ip 0 x_p\wedge \neg(x_i\ip 0 y_p)) & \{\ip 0\}\\
                                   & \\
                                  
                                   \exists x_iy_i(x_i\ip1 x_p\wedge y_i\ip 1 y_p\wedge x_i\ii 24 y_i) & \{\ip 1,\ii 24\}\\
                                   & \\
                                   \exists x_iy_i(x_i\ip2 x_p\wedge y_i\ip 2 y_p\wedge \neg(x_i\ip 2 y_p)\wedge \neg(y_i\ip 2 x_p)\wedge x_i\ii 24 y_i) & \{\ip 2,\ii 24\}\\
                                   & \\
                                   \exists x_iy_i(x_i\ip2 x_p\wedge y_i\ip 2 y_p \wedge \neg(x_i\ip 2 y_p)\wedge x_i\ii 14 y_i) & \{\ip 2,\ii 14\}\\
                                   &\\
                                   \exists x_iy_i(x_i\ip2 x_p\wedge y_i\ip 2 y_p \wedge \neg(x_i\ip 2 y_p)\wedge\neg(y_i\ip 2 x_p)\wedge  & \{\ip 2,\ii 44\}\\
                                   \hspace{0.9cm}\exists z_i,t_i(z_i\ii 44 y_i\wedge\neg(z_i\ii 44 x_i)\wedge x_i\ii44 t_i\wedge \neg(y_i\ii 44 t_i)))\\ 
\end{array} \right. \]
\end{small}

\medskip

\noindent Consider, first the set \underline{$\{\ip 0\}$}. Suppose that $\mathcal F\models\varphi(a,b)$, and, for the sake of contradiction, that $a\ge b$. If $a=b$, then every interval $[c,d]$ such that $[c,d]\ip 0 a$ must be such that $[c,d]\ip 0 b$ as well, contradicting the second conjunct, and, if $b<a$ then, by unboundedness, there exists an interval $[a,c]$ such that $[a,c]\ip 0 b$ but it is not the case that $[a,c]\ip 0 a$, contradicting the first conjunct. On the other hand, suppose that $x_p = a < b$. Then if $[c,d]\ip 0 b$, we have $b<c$ and hence $a<c$, therefore $[c,d]\ip 0 a$, satisfying the first conjunct. For the sake of the second conjunct, consider any interval $[b,c]$: such an interval exists by unboundedness and clearly $[b,c]\ip 0 a$ while it is not the case that $[b,c]\ip 0 b$. Consider, now, the set \underline{$\{\ip 1,\ii 24\}$}. If $x_p = a < b = y_p$ then, by unboundedness, we may choose points $c$ and $d$ such that $a < b < c < d$. Then the intervals $x_i = [a,c]$ and $y_i = [b,d]$ witness $\varphi$. Conversely, if $\mathcal F\models\varphi(a,b)$, then there must be intervals $[a,c]$ and $[b,d]$ such that $[a,c] \ii24 [b,d]$, i.e.\ such that $a<b<c<d$. So, in particular, $a < b$. As for \underline{$\{\ip 2,\ii 24\}$}, suppose that $\mathcal F\models\varphi(a,b)$, and, by contradiction, that $a\ge b$. If $a=b$, then for any interval {\em containing} $a$ also {\em contains} $b$, and vice versa, making the first three conjuncts not simultaneously satisfiable. If, on the other hand, $b<a$, any interval {\em containing} $b$ must {\em start} before $a$, and therefore no interval {\em containing} $a$ can {\em overlap} an interval {\em containing} $b$, making $\varphi(a,b)$ false. Conversely, assume that $a<b$. By unboundedness we can take two intervals $[c,b]$ such that $c<a$ and $[a,d]$ such that $b<d$ to witness $\varphi(a,b)$, as we wanted. As for \underline{$\{\ip 2,\ii 14\}$}, suppose that $\mathcal F\models\varphi(a,b)$, and, by contradiction, that $a\ge b$. If $a=b$, then $a$ and $b$ are contained by the same intervals, making the first three conjuncts not simultaneously satisfiable. If, on the other hand, $b<a$, then any interval {\em containing} $a$ but not $b$ must {\em start} at $b$ or after, and therefore it cannot {\em start} an interval containing $b$, falsifying $\varphi(a,b)$. Conversely, suppose that $a<b$: then, the intervals $[a,b]$ and $[a,c]$ for some $c>b$ witness $\varphi(a,b)$, as we wanted. Finally, as for \underline{$\{\ip 2,\ii 44\}$}, suppose that $\mathcal F\models\varphi(a,b)$, and, by contradiction, that $a\ge b$. If $a=b$, then $a$ and $b$ are contained by the same intervals, making the first four conjuncts not simultaneously satisfiable. If, on the other hand, $b<a$, then any interval {\em containing} $a$ but not $b$ must {\em start} at $b$ or after it, and every interval {\em containing} $b$ but not $a$ must {\em end} at $a$ or before it; thus, $z_i$ and $t_i$ cannot be witnessed by any concrete interval, as $z_i$ should {\em end} before $y_i$ {\em starts} but not before $x_i$ does, and, symmetrically, $x_i$ should {\em end} before $t_i$ {\em starts} but not before $y_i$ does. Conversely, suppose that $a<b$: then, any interval of the type $[c,b]$ ($c<a$) serves as witness of $x_i$, and any interval of the type $[a,d]$ ($d>b$) serves as witness of $y_i$, while $z_i$ and $t_i$ are witnessed by any intervals $[e,c]$ ($e < c$) and $[d,f]$ ($d < f$), respectively. Such points $c$, $d$, $e$ and $f$ exist thanks to the unboundedness of the order. Every other definition is now straightforward. The cases for $\underline{\{\ip 2, \ii 34\}}$ follows since $\{\ii 34\}$ defines $\ii 14$ in the general case and the just-proved $<$-completeness of $\{\ip 2,\ii 14\}$. The cases for $\underline{\{\ip 1, \ii 04\}}$ and $\underline{\{\ip 1, \ii 44\}}$ follow since both $\{\ii 04\}$ and $\{\ii44\}$ define $=_i$ in the unbounded case and $\{\ip 1, \ii 04,  =_i\}$ define $<$ in the general case. The remaining cases, viz.\  $\underline{\{\ip 2, \ii 03\}}$, $\underline{\{\ip 3, \ii 24\}}$, $\underline{\{\ip 3, \ii 04\}}$, $\underline{\{\ip 3, \ii 44\}}$  and $\underline{\{\ip 4\}}$, now follow by symmetry and Lemma \ref{lem:symm} from the already-proved cases for $\{\ip 2, \ii 14 \}$, $\{\ip 1, \ii 24 \}$, $\{\ip 1, \ii 04 \}$, $\{\ip 1, \ii 44 \}$ and $\{\ip 0\}$, respectively.
\end{proof}

\subsection{Definability for $\allm^+$-Relations in $\Unb$}

Because the relationship between definability of $\allm^+$ and $\alli^+$ over unbounded orders is somewhat involved, the corresponding results cannot be kept as neatly apart as e.g.\ in the dense case. We therefore first present and prove two lemmas (Lemmas \ref{lem:first:allm:comp:unb} and \ref{lem:unb:anticipate}) which collect some key definability results for $\allm^+$ and $\alli^+$, respectively. These results are then used to prove the full collection of minimal definability results for $\allm^+$ and $\alli^+$ (Lemmas \ref{lem:allm:unb} and \ref{lem:alli:unb}).

\begin{lemma}\label{lem:first:allm:comp:unb}
	Over unbounded linear orders the sets $\{\ip 2,\ii 04, <\}$ and $\{\ip 1, \ii 24\}$ are $\ip 0$-complete, the set $\{\ip 0\}$ is $\ip 1$-complete, while the sets $\{\ip 1,\ii 24\}$, $\{\ip 1,\ii 04\}$, $\{\ip 0,\ii 24\}$ and $\{\ip 0,\ii 04\}$ are $\ip 2$-complete.
\end{lemma}

\begin{proof}
	Consider the following definition of $\ip 0$ in terms of $\ip 2$,$\ii 04$ and $<$:
	
	\medskip
	
	\begin{small}
		\[
		\begin{array}{llll}
		x_i\ip 0 y_p &\leftrightarrow &\exists z_i(x_i\ii 04 z_i \wedge \forall k_p(z_i\ip 2 k_p\rightarrow y_p < k_p)). & \{\ip 2,\ii 04, <\}
		\end{array}
		\]
	\end{small}
	
	\medskip
	
	\noindent Proceeding as always, suppose that $\mathcal F\models\varphi([a,b],c)$. Let $z_i = [d,e]$ the interval witnessing $\varphi$, so that  $d < a < b < e$. Since $z_i \ip 2 a$ by definition it follows that $c < a$, i.e.\ that $[a,b] \ip 0 c$. Conversely, suppose that $[a,b] \ip 0 c$. Since the underlying order is unbounded there is a point $d$ such that $b < d$. Then the interval $[c,d]$ witnesses the definition. As for $\ip 1$, we show that \underline{$\{\ip 0\}$} is $\ip 1$-complete by noticing that \underline{$\{\ip 0\}$} is $<,=_p$-complete and by using the following, straightforward definition:
	
	\medskip
	
	\begin{small}
		\[
		\begin{array}{llll}
		x_i\ip 1 y_p &\leftrightarrow &\forall z_p(x_i\ip 0 z_p\rightarrow z_p<y_p)\wedge\forall z_p(\neg(x_i\ip 0 z_p)\rightarrow (y_p<z_p\vee y_p=_p z_p)). & \{\ip 0\}
		\end{array}
		\]
	\end{small}
	
	\medskip
	
	\noindent Finally, consider the following two new definitions for $\ip 2$:
	
	\medskip
	
	\begin{small}
		\[ x_i \ip 2 y_p \leftrightarrow \left\{ \begin{array}{ll}
		\exists z_i(z_i\ip 1 y_p\wedge x_i\ii 24 z_i) & \{\ip 1,\ii 24\} \\
		& \\
		\neg(x_i\ip 1 y_p)\wedge\exists z_i(z_i\ip 1 y_p\wedge\forall k_i(x_i\ii 04 k_i\rightarrow z_i\ii 04 k_i)). & \{\ip 1,\ii 04\} \\
		\end{array} \right. \]
	\end{small}
	
	\medskip
	
	\noindent The first one, namely \underline{$\{\ip 1,\ii 24\}$}, is straightforward. As for \underline{$\{\ip 1,\ii 04\}$},
	suppose first that $\mathcal F\models\varphi([a,b],c)$. If $c<a$, then we have the interval $[c,a]$ that contradicts the definition; $c$ cannot be $a$, and if $c\ge b$, then every interval $[c,d]$, where $d>c\ge b$ exists by hypothesis, again, leads to a contradiction. Thus, $a<c<b$. If, on the other hand, $a<c<b$ then $z_i=[c,b]$ witnesses the existential quantifier in the definition. The remaining sets are $\{\ip 1, \ii 24\}$, $\{\ip 0,\ii 24\}$ and $\{\ip 0,\ii 04\}$. These do not require new definitions. Indeed, the $\ip 0$-completeness of $\underline{\{\ip 1, \ii 24\}}$ follows since $\{\ip 1, \ii 24\}$ defines $\ip 2$ over unbounded orders (as we proved above) and $\{\ip 1, \ip 2 \}$ defines  $\ip 0$ over all linear orders. That $\underline{\{\ip 0,\ii 24\}}$ and $\underline{\{\ip 0,\ii 04\}}$ are $\ip 2$-complete follows since $\{\ip 0 \}$ defines $\ip 1$ over unbounded order while, as we have just proved, $\{\ip 1,\ii 24\}$ and $\{\ip 1,\ii 04\}$ both define $\ip 2$ over these order.
\end{proof}

\begin{lemma}\label{lem:unb:anticipate}
	The sets $\{\ip 0, \ii 24 \}$, $\{\ip 1, \ii 24 \}$, $\{\ip 0, \ii 04 \}$, $\{\ip 1, \ii 04 \}$,  $\{\ip 4, \ii 24 \}$, $\{\ip 3, \ii 24 \}$, $\{\ip 4, \ii 04 \}$,  $\{\ip 3, \ii 04 \}$ are $\ii 34$-complete over unbounded linear orders.
\end{lemma}

\begin{proof}
	No new definitions are required, as all these definability results follow from previous results. Indeed, we begin by recalling that $\{\ip 0, \ip 2\}$ defines $\ii 34$ over all linear orders. From this the cases of $\underline{\{\ip 0, \ii 24 \}}$ and $\underline{\{\ip 0, \ii 04 \}}$ follow via the fact that both of them define $\ip 2$ over unbounded orders (Lemma \ref{lem:first:allm:comp:unb}). The cases for $\underline{\{\ip 1, \ii 24 \}}$ and $\underline{\{\ip 1, \ii 04 \}}$ follow similarly since both define $\ip 2$ over unbounded orders (Lemma \ref{lem:first:allm:comp:unb}) and $\{\ip 1, \ip 2\}$ defines $\ii 34$ over all linear orders. Lastly, the cases for  $\underline{\{\ip 4, \ii 24 \}}$, $\underline{\{\ip 3, \ii 24 \}}$, $\underline{\{\ip 4, \ii 04 \}}$,  $\underline{\{\ip 3, \ii 04 \}}$ follow by symmetry (Lemma \ref{lem:symm}) from the previous ones.	
\end{proof}

\begin{table}[t]
	\small
	\begin{center}
		\begin{tabular}{|p{0.20\textwidth}|p{0.20\textwidth}|p{0.20\textwidth}|}
		\hline
		$\ip 0$                                & $\ip 1$                                & $\ip 2$                          \\ \hline
		$\{\ip1, \ip2 \}$                      & $\boldsymbol{\{\ip0 \}}$               & $\{\ip0, \ip3 \}$                \\
		$\{\ip1, \ip3 \}$                      & $\{\ip2, \ip3 \}$                      & $\{\ip0, \ip4 \}$                \\
		$\{\ip1, \ip4 \}$                      & $\boldsymbol{\{\ip2, \ip4 \}}$         & $\boldsymbol{\{\ip0, \ii24\}}$   \\
		$\boldsymbol{\{\ip1, \ii24 \}}$        & $\boldsymbol{\{\ip2, \ii14 \}}$        & $\boldsymbol{\{\ip0,  \ii44 \}}$ \\
		$\{\ip1, \ii03 \}$                     & $\boldsymbol{\{\ip2, \ii24, \ii04 \}}$ & $\boldsymbol{\{\ip0, \ii03\}}$   \\
		$\{\ip1, \ii34 \}$                     & $\boldsymbol{\{\ip2,  \ii44 \}}$       & $\boldsymbol{\{\ip0, \ii34 \}}$  \\
		$\boldsymbol{\{\ip1, \ii04 \}}$        & $\boldsymbol{\{\ip2, \ii03 \}}$        & $\boldsymbol{\{\ip0, \ii04 \}}$  \\
		$\boldsymbol{\{\ip1, \ii44 \}}$        & $\boldsymbol{\{\ip2, \ii34 \}}$        & $\{\ip1, \ip3 \}$                \\
		$\{\ip1, < \}$                         & $\boldsymbol{\{\ip2, \ii04, < \}}$     & $\{\ip1, \ip4 \}$                \\
		$\{\ip2, \ip3 \}$                      & $\{\ip3, \ii14 \}$                     & $\boldsymbol{\{\ip1, \ii24 \}}$  \\
		$\{\ip2, \ip4 \}$                      & $\boldsymbol{\{\ip3, \ii24 \}}$        & $\boldsymbol{\{\ip1,  \ii44 \}}$ \\
		$\boldsymbol{\{\ip2, \ii14 \}}$        & $\boldsymbol{\{\ip3,  \ii04 \}}$       & $\{\ip1, \ii03 \}$               \\
		$\boldsymbol{\{\ip2, \ii24, \ii04 \}}$ & $\boldsymbol{\{\ip3,  \ii44 \}}$       & $\{\ip1, \ii34 \}$               \\
		$\boldsymbol{\{\ip2,  \ii44 \}}$       & $\{\ip3, \ii34 \}$                     & $\boldsymbol{\{\ip1, \ii04 \}}$  \\
		$\boldsymbol{\{\ip2, \ii03 \}}$        & $\boldsymbol{\{\ip4, \ii14 \}}$        & $\{\ip3, \ii14 \}$               \\
		$\{\ip2, \ii34 \}$                     & $\boldsymbol{\{\ip4, \ii24 \}}$        & $\boldsymbol{\{\ip3, \ii24 \}}$  \\
		$\boldsymbol{\{\ip2, \ii04, < \}}$     & $\boldsymbol{\{\ip4,  \ii04 \}}$       & $\{\ip3, \ii04 \}$               \\
		$\{\ip3, \ii14 \}$                     & $\boldsymbol{\{\ip4,  \ii44 \}}$       & $\boldsymbol{\{\ip3,  \ii44 \}}$ \\
		$\boldsymbol{\{\ip3, \ii24 \}}$        & $\boldsymbol{\{\ip4, \ii34 \}}$        & $\{\ip3, \ii34 \}$               \\
		$\boldsymbol{\{\ip3, \ii04 \}}$        &                                        & $\boldsymbol{\{\ip4, \ii14 \}}$  \\
		$\boldsymbol{\{\ip3, \ii44 \}}$        &                                        & $\boldsymbol{\{\ip4, \ii24 \}}$  \\
		$\{\ip3, \ii34 \}$                     &                                        & $\boldsymbol{\{\ip4, \ii04 \}}$  \\
		$\boldsymbol{\{\ip4, \ii14 \}}$        &                                        & $\boldsymbol{\{\ip4,  \ii44 \}}$ \\
		$\boldsymbol{\{\ip4, \ii24 \}}$        &                                        & $\boldsymbol{\{\ip4, \ii34 \}}$  \\
		$\boldsymbol{\{\ip4,  \ii04 \}}$       &                                        &                                  \\
		$\boldsymbol{\{\ip4,  \ii44 \}}$       &                                        &                                  \\
		$\boldsymbol{\{\ip4, \ii34 \}}$        &                                        &                                  \\
		&                                        &                                  \\ \hline
	\end{tabular}
\caption{The spectrum of the $\mathsf{mcs}(r)$, for each $r\in\allm^+$. - Class: $\Unb$.}\label{tab:mixedunb}
\end{center}
\end{table}

\begin{lemma}\label{lem:allm:unb}
Tab.~\ref{tab:mixedunb} is correct.
\end{lemma}

\begin{proof}
No new explicit definitions are needed. We focus first on $\ip 0$-completeness. The cases for $\underline{\{\ip 2,\ii 04, <\}}$ and $\underline{\{\ip 1, \ii 24\}}$ were already dealt with in Lemma \ref{lem:first:allm:comp:unb}.
%For the  $\ip 0$-completeness of $\underline{\{\ip 1, \ii 24\}}$ we note that, by Lemma \ref{lem:unb:anticipate} we have $\{\ip 1, \ii 24\}$ defines $\ii 34$ in the unbounded case, and than $\{\ip 1, \ii 34\}$ defines $\ip 0$.
In the cases $\underline{\{ \ip 1, \ii 04 \}}$ and  $\underline{\{ \ip 1, \ii 44 \}}$ we note that both  $\{\ii 04 \}$ and $\{\ii 44 \}$ define $=_i$ in the unbounded case, and that both $\{ \ip 1, \ii 04, =_i \}$ and $\{ \ip 1, \ii 44, =_i \}$ and $\ip 0$-complete. For $\underline{\{ \ip 2, \ii 14 \}}$, note that $\{ \ip 2, \ii 14 \}$ defines $<$ in the unbounded case and that $\{ \ip 2, \ii 14, < \}$ defines  $\ip 0$. The case for $\underline{\{\ip 2, \ii 24, \ii 04 \}}$ follows from the facts that, in the unbounded case, $\{\ip 2, \ii 24 \}$ defines $<$ and $\{\ip 2, \ii 04, < \}$ defines $\ip 0$, as already proved. The case for $\underline{\{\ip 2,  \ii 44 \}}$ now follows from the latter via the fact that $\{\ii 44 \}$ is $\ii 24$-complete and $\{\ii 44 \}$ is $\ii 04$-compete in the unbounded case. For $\underline{\{\ip 2,  \ii 03 \}}$, observe that  $\{\ip 2, \ii 03 \}$ defines $<$ in the unbounded case and that $\{\ip 2, \ii 03, < \}$ is $\ip 0$-complete. To see that $\underline{\{\ip 3,  \ii 24 \}}$ is $\ip 0$-complete, one may use the facts that $\{\ip 1, \ii 24\}$ defines $\ii 34$ in the unbounded case(Lemma \ref{lem:unb:anticipate}), hence, by symmetry (Lemma~\ref{lem:symm}), $\{\ip 3, \ii 24\}$ defines $ \ii 34$, and that $\{\ip 3, \ii 34\}$ defines  $\ip 0$ in the general case. For $\underline{\{\ip 3,  \ii 04 \}}$, note that $\{\ip 3, \ii 04\}$ defines $\ii 34$ (Lemma \ref{lem:unb:anticipate}) in the unbounded case, and then again that $\{\ip 3, \ii 34\}$ defines  $\ip 0$ in the general case. The cases for $\underline{\{\ip 3,  \ii 24 \}}$,  $\underline{\{\ip 3,  \ii 04 \}}$ and  $\underline{\{\ip 3,  \ii 44 \}}$ all follow from the fact that $\{\ip 3, \ii 34\}$ defines $\ip 0$, via that fact that each of these three sets defines $\ii 34$ over $\Unb$: in the first instance $\{\ip 1, \ii 24\}$ defines $\ii 34$ (Lemma \ref{lem:unb:anticipate}), hence, by symmetry (Lemma~\ref{lem:symm}) so does $\{\ip 3, \ii 24\}$; in the second instance we have that $\{\ip 3, \ii 04\}$ defines  $\ii 34$ (Lemma \ref{lem:unb:anticipate}); in the third one we have that $\ii 44$ defines, directly, $\ii 34$ in the unbounded case. The set $\underline{\{\ip 4,  \ii 14 \}}$ is $\ip 0$-complete since $\{\ip 4 \}$ defines $<$ in the unbounded case and $\{\ip 4,  \ii 14, < \}$ defines $\ip 0$.  The $\ip 0$-completeness of $\underline{\{\ip 4,  \ii 34 \}}$ follows from that of $\{\ip 4,  \ii 14 \}$ since $\ii 34$ defines all interval-interval relations over all linear orders. From this, in turn, follow the cases for  $\underline{\{\ip 4,  \ii 24 \}}$ and  $\underline{\{\ip 4,  \ii 44 \}}$ since both these sets define $\ii 34$ over unbounded orders (Lemma \ref{lem:first:allm:comp:unb}). Lastly the case for $\underline{\{\ip 4,  \ii 44 \}}$ also follows from the $\ip 0$-completeness of $\{\ip 4,  \ii 14 \}$ since $\ii 44$ defines $\ii 14$ over unbounded orders. We now turn our attention to $\ip 1$-completeness. The case of \underline{$\{\ip 0\}$} was already treated in Lemma \ref{lem:first:allm:comp:unb}. All remaining cases follow immediately from this, as the sets involved have already been shown to define $\ip 0$ over unbounded orders. Now focusing on $\ip 2$-completeness, the cases for $\underline{\{\ip 1,\ii 24\}}$, $\underline{\{\ip 1,\ii 04\}}$, $\underline{\{\ip 0,\ii 24\}}$ and $\underline{\{\ip 0,\ii 04\}}$ were already considered in Lemma \ref{lem:first:allm:comp:unb}. The cases for $\underline{\{\ip 0,\ii 44\}}$, $\underline{\{\ip 0,\ii 03\}}$ and $\underline{\{\ip 0,\ii 34\}}$ all follow from the $\ip 2$-completeness of $\{\ip 0,\ii 24\}$ via, respectively, the facts that $\{\ii 44 \}$ defines $\ii 24$ over unbounded orders,  $\{\ip 0,\ii 03\}$ defines $\ii 24$ over linear orders, and $\{\ii 34 \}$ defines $\ii 24$ over unbounded orders. The $\ip 2$-completeness of $\underline{\{\ip 1,\ii 44\}}$ follows from that of  $\{\ip 1,\ii 24\}$, since $\{\ii 44\}$ defines $\ii 24$ over unbounded orders. The cases for $\underline{\{\ip 3,\ii 24\}}$ and $\underline{\{\ip 3,\ii 44\}}$ follow from the fact that both define $\ip 0$ over unbounded orders and the $\ip 2$ completeness of, respectively, $\{\ip 0,\ii 24\}$ and $\{\ip 0,\ii 44\}$. The $\ip 2$ completeness of $\underline{\{\ip 4,\ii 14\}}$ follows by symmetry (Lemma \ref{lem:symm}) from that of $\{\ip 0,\ii 03\}$. To see that $\underline{\{\ip 4,\ii 24\}}$, $\underline{\{\ip 4,\ii 04\}}$, $\underline{\{\ip 4,\ii 44\}}$ and $\underline{\{\ip 4,\ii 34\}}$ we note that each of these sets defines $\ip 0$ over unbounded orders and than appeal to the $\ip 2$-completeness of the sets $\{\ip 0,\ii 24\}$, $\{\ip 0,\ii 04\}$, $\{\ip 0,\ii 44\}$ and $\{\ip 0,\ii 34\}$, respectively.
\end{proof}

\subsection{Definability for Relations in $\alli^+\setminus\{=_i\}$ in $\Unb$}\label{Sec:allI:unb}

\begin{table}[t]
\begin{center}
\small
\begin{tabular}{|p{0.17\textwidth}|p{0.17\textwidth}|p{0.17\textwidth}|p{0.17\textwidth}|p{0.17\textwidth}|}
	\hline
	$\ii 14$                               & $\ii 34$                           & $\ii 24$                        & $\ii 04$                        & $\ii44$                            \\ \hline
	$\{\ip0, \ip2 \}$                      & $\{\ip0, \ip2 \}$                  & $\{\ip0, \ip2 \}$               & $\{\ip0, \ip2 \}$               & $\{\ip0, \ip2 \}$                  \\
	$\{\ip0, \ip3 \}$                      & $\{\ip0, \ip3 \}$                  & $\{\ip0, \ip3 \}$               & $\{\ip0, \ip3 \}$               & $\{\ip0, \ip3 \}$                  \\
	$\{\ip0, \ip4 \}$                      & $\{\ip0, \ip4 \}$                  & $\{\ip0, \ip4 \}$               & $\{\ip0, \ip4 \}$               & $\{\ip0, \ip4 \}$                  \\
	$\boldsymbol{\{\ip0, \ii24 \}}$        & $\boldsymbol{\{\ip0, \ii24 \}}$    & $\{\ip0, \ii03 \}$              & $\boldsymbol{\{\ip0, \ii24 \}}$ & $\boldsymbol{\{\ip0, \ii24 \}}$    \\
	$\{\ip0, \ii03 \}$                     & $\{\ip0, \ii03 \}$                 & $\boldsymbol{\{\ip0, \ii04 \}}$ & $\{\ip0, \ii03 \}$              & $\{\ip0, \ii03 \}$                 \\
	$\boldsymbol{\{\ip0, \ii04 \}}$        & $\boldsymbol{\{\ip0, \ii04 \}}$    & $\{\ip1, \ip2 \}$               & $\{\ip1, \ip2 \}$               & $\boldsymbol{\{\ip0, \ii04 \}}$    \\
	$\{\ip1, \ip2 \}$                      & $\{\ip1, \ip2 \}$                  & $\{\ip1, \ip3 \}$               & $\{\ip1, \ip3 \}$               & $\{\ip1, \ip2 \}$                  \\
	$\{\ip1, \ip3 \}$                      & $\{\ip1, \ip3 \}$                  & $\{\ip1, \ip4 \}$               & $\{\ip1, \ip4 \}$               & $\{\ip1, \ip3 \}$                  \\
	$\{\ip1, \ip4 \}$                      & $\{\ip1, \ip4 \}$                  & $\{\ip1, \ii03 \}$              & $\boldsymbol{\{\ip1, \ii24 \}}$ & $\{\ip1, \ip4 \}$                  \\
	$\boldsymbol{\{\ip1, \ii24 \}}$        & $\boldsymbol{\{\ip1, \ii24 \}}$    & $\boldsymbol{\{\ip1, \ii04 \}}$ & $\{\ip1, \ii03 \}$              & $\boldsymbol{\{\ip1, \ii24 \}}$    \\
	$\{\ip1, \ii03 \}$                     & $\{\ip1, \ii03 \}$                 & $\{\ip2, \ip3 \}$               & $\{\ip2, \ip3 \}$               & $\{\ip1, \ii03 \}$                 \\
	$\boldsymbol{\{\ip1, \ii04 \}}$        & $\boldsymbol{\{\ip1, \ii04 \}}$    & $\{\ip2, \ip4 \}$               & $\{\ip2, \ip4 \}$               & $\boldsymbol{\{\ip1, \ii04 \}}$    \\
	$\{\ip2, \ip3 \}$                      & $\{\ip2, \ip3 \}$                  & $\{\ip2, \ii14 \}$              & $\{\ip2, \ii14 \}$              & $\{\ip2, \ip3 \}$                  \\
	$\{\ip2, \ip4 \}$                      & $\{\ip2, \ip4 \}$                  & $\{\ip2, \ii03 \}$              & $\{\ip2, \ii03 \}$              & $\{\ip2, \ip4 \}$                  \\
	$\boldsymbol{\{\ii24, \ii04 \}}$ & $\{\ip2, \ii14 \}$    & $\{\ip2, < \}$                  & $\{\ip3, \ii14 \}$              & $\{\ip2, \ii14 \}$                 \\
	$\{\ip2, \ii03 \}$                     & $\{\ip2, \ii03 \}$                 & $\{\ip3, \ii14 \}$              & $\{\ip4, \ii14 \}$              & $\{\ip2, \ii03 \}$                 \\
	$\boldsymbol{\{\ip2, \ii04, < \}}$     & $\boldsymbol{\{\ip2, \ii04, < \}}$ & $\{\ip4, \ii14 \}$              & $\{\ii14, \ii24 \}$             & $\boldsymbol{\{\ip2, \ii04, < \}}$ \\
	$\{\ii24, \ii03 \}$                    & $\{\ip3, \ii14 \}$                 & $\{\ii14, \ii03 \}$             & $\{\ii14, \ii03 \}$             & $\{\ip3, \ii14 \}$                 \\
	$\{\ii03, \ii04 \}$                    & $\boldsymbol{\{\ip4, \ii14 \}}$    & $\{\ii14, \ii04 \}$             & $\{\ii24, \ii03 \}$             & $\{\ip4, \ii14 \}$                 \\
	$\{\ii34 \}$                           & $\{\ii14, \ii24 \}$                & $\{\ii03, \ii04 \}$             & $\{\ii34 \}$                    & $\{\ii14, \ii24 \}$                \\
	$\boldsymbol{\{\ii44 \}}$              & $\{\ii14, \ii03 \}$                & $\{\ii34 \}$                    & $\boldsymbol{\{\ii44 \}}$       & $\{\ii14, \ii03 \}$                \\
	$\boldsymbol{\{\ip 3, \ii24 \}}$       & $\{\ii14, \ii04 \}$                & $\boldsymbol{\{\ii44 \}}$       &$\boldsymbol{\{\ip 3, \ii24 \}}$                                 & $\{\ii14, \ii04 \}$                \\
	$\boldsymbol{\{\ip 4, \ii24 \}}$       & $\{\ii24, \ii03 \}$                & $\boldsymbol{\{\ip 3, \ii04 \}}$                                & $\boldsymbol{\{\ip 4, \ii24 \}}$                                 & $\{\ii24, \ii03 \}$                \\
	$\boldsymbol{\{\ip 3, \ii04 \}}$    & $\boldsymbol{\{\ii24, \ii04 \}}$   & $\boldsymbol{\{\ip 4, \ii04 \}}$                                  &                                 & $\boldsymbol{\{\ii24, \ii04 \}}$   \\
	$\boldsymbol{\{\ip 4, \ii04 \}}$  & $\{\ii03, \ii04 \}$                &                                 &                                 & $\{\ii03, \ii04 \}$                \\
	                                       & $\boldsymbol{\{\ii44 \}}$          &                                 &                                 & $\{\ii34 \}$                       \\
	                                       & $\boldsymbol{\{\ip 1, \ii24 \}}$ & & &$\boldsymbol{\{\ip 3, \ii24 \}}$\\
	                                       & $\boldsymbol{\{\ip 3, \ii24 \}}$ & & &$\boldsymbol{\{\ip 4, \ii24 \}}$\\
	                                       & $\boldsymbol{\{\ip 4, \ii24 \}}$ & & &$\boldsymbol{\{\ip 3, \ii04 \}}$\\
	                                       & $\boldsymbol{\{\ip 3, \ii04 \}}$ & & &$\boldsymbol{\{\ip 4, \ii04 \}}$\\
	                                       & $\boldsymbol{\{\ip 4, \ii04 \}}$ & & &\\ \hline
\end{tabular}
\end{center}\caption{The spectrum of the $\mathsf{mcs}(r)$, for each $r\in\alli^+\setminus\{=_i\}$. - Class: $\Unb$.}\label{tab:iunb}
\end{table}

\begin{lemma}\label{lem:alli:unb}
Tab.~\ref{tab:iunb} is correct.
\end{lemma}

\begin{proof}
We begin by noticing that every new $\ii 14$-, $\ii 24$-, $\ii 04$- or $\ii 44$-complete set also appears as a $\ii 34$-complete set. Since $\ii 34$ defines every other interval-interval relation over linear orders, it is therefore sufficient to justify the $\ii 34$-completeness of all new $\ii 34$-complete sets. This has already been done for  $\underline{\{\ii 24, \ii 04 \}}$ and  $\underline{\{\ii 44 \}}$ in Lemma \ref{lem:Pure:Int:Inte:Und}, and for $\underline{\{\ip 0, \ii 24 \}}$, $\underline{\{\ip 1, \ii 24 \}}$, $\underline{\{\ip 0, \ii 04 \}}$, $\underline{\{\ip 1, \ii 04 \}}$,  $\underline{\{\ip 4, \ii 24 \}}$, $\underline{\{\ip 3, \ii 24 \}}$, $\underline{\{\ip 4, \ii 04 \}}$,  $\underline{\{\ip 3, \ii 04 \}}$ in Lemma \ref{lem:unb:anticipate}. The $\ii 34$-completeness of $\underline{\{\ip 1, \ii 24 \}}$ follows from the the $\ip 2$-completeness of $\{\ip 1, \ii 24 \}$ over unbounded orders and the  $\ii 34$-completeness of $\{\ip 1, \ip 2 \}$ over linear orders. The case for $\underline{\{\ip 2, \ii 04, < \}}$ follows since this set defines $\ip 1$ over unbounded orders and $\{\ip 1, \ip 2\}$ is $\ii 34$-complete over linear orders. Lastly, the $\ii 34$-completes of  $\underline{\{\ip 4, \ii 14\}}$ is a consequence of the $\ip 1$- and $\ip 2$-completeness of this set over unbounded order and the $\ii 34$-completes of $\{\ip 1, \ip 2 \}$ over linear orders.
\end{proof}

\section{Incompleteness Results in The Class $\Unb$}\label{sec:unb-incomp}

We can now turn our attention to the maximal incomplete sets for relations in $\allr^+$. Notice that for some $r\in\allr^+$, some $r$-incomplete
set in the class $\Den$ is also maximally $r$-incomplete in the class $\Unb$, and it has been proven so by means of a dense unbounded
counterexample; in these cases, we can borrow the same argument unchanged.

\begin{table}[t]
\small
\begin{center}
\begin{tabular}{p{0.34\textwidth}|p{0.022\textwidth}|p{0.022\textwidth}|p{0.022\textwidth}|p{0.022\textwidth}|p{0.022\textwidth}|p{0.022\textwidth}|p{0.022\textwidth}|p{0.022\textwidth}|p{0.022\textwidth}|p{0.022\textwidth}|p{0.022\textwidth}|p{0.022\textwidth}|p{0.022\textwidth}|p{0.022\textwidth}|p{0.022\textwidth}|}
{\em Proved}                            & $=_p$             & $=_i$             & $<$              & $\ip0$           & $\ip1$           & $\ip2$          & $\ip3$           & $\ip4$           & $\ii34$          & $\ii14$          & $\ii03$          & $\ii24$          & $\ii04$          & $\ii44$           \\
 \hline
$\alli^{+}$                             & $\bullet$         &                   &                  &                  &                  &                 &                  &                  &                  &                  &                  &                  &                  &                   \\
 \hline
$\{=_p, <, \ip 0,\ip 1\}$               &                   & $\bullet$         &                  &                  &                  &                 &                  &                  &                  &                  &                  &                  &                  &                   \\
\hline
$\{=_p, <, \ip2, \ii 24\}$              &                   & $\bullet$         &                  &                  &                  &                 &                  &                  &                  &                  &                  &                  &                  &                   \\
\hline
$\{=_p, =_i, \ip1, \ii 14\}$            &                   &                   & $\bullet$        & $\bullet$        &                  &                 &                  &                  &                  &                  &                  &                  &                  &                   \\
\hline
$\{=_p\}\cup\alli^+$                    &                   &                   & $\bullet$        &         &         & &         &         &                  &                  &                  &                  &                  &                   \\
\hline
$\{=_p,<\}\cup\alli^+$                    &                   &                   &        &  $\bullet$        & $\bullet$         &$\bullet$  & $\bullet$         & $\bullet$         &                  &                  &                  &                  &                  &                   \\
\hline
$\{=_p, =_i, \ip 2, \ii 04\}$           &                   &                   & $\bullet$        & $\bullet$        & $\bullet$        &                 & $\bullet$        & $\bullet$        & $\bullet$        & $\bullet$        & $\bullet$        & $\bullet$        &                  &  $\bullet$        \\
\hline
$\{=_p, =_i,<,\ip 3,\ip4, \ii03\}$      &                   &                   &                  & $\bullet$        & $\bullet$        & $\bullet$       &                  &                  & $\bullet$        & $\bullet$        &                  & $\bullet$        & $\bullet$        &  $\bullet$        \\
\hline
$\{=_p, =_i,<,\ip 2,\ii24\}$            &                   &                   &                  & $\bullet$        & $\bullet$        &                 & $\bullet$        & $\bullet$        & $\bullet$        & $\bullet$        & $\bullet$        &                  & $\bullet$        &  $\bullet$        \\
\hline
$\{=_p, =_i,<,\ip 0,\ip 1\}$            &                   &                   &                  &                  &                  &                 &                  &                  &                  & $\bullet$        &                  &                  &                  &               \\
\hline
$\{=_p, =_i,<,\ii 04\}$                 &                   &                   &                  &                  &                  &                 &                  &                  & $\bullet$        & $\bullet$        & $\bullet$        & $\bullet$        &                  &  $\bullet$    \\
\hline

\\
{\em Symmetric}                         & $=_p$             & $=_i$             & $<$              & $\ip0$           & $\ip1$           & $\ip2$          & $\ip3$           & $\ip4$           & $\ii34$          & $\ii14$          & $\ii03$          & $\ii24$          & $\ii04$          & $\ii44$           \\
\hline
$\{=_p, <, \ip 3,\ip 4\}$                &                   &$\bullet$           &                  &                  &                  &                     &                  &                  &                  &                  &                  &                  &                  &                   \\
\hline
$\{=_p, =_i, \ip 3, \ii 03\}$            &                   &                    &$\bullet$         &                  &                  &                     &                  & $\bullet$        &                  &                  &                  &                  &                  &                   \\
\hline
$\{=_p, =_i, <,\ip 0, \ip 1, \ii 14\}$   &                   &                    &                  &                  &                  & $\bullet$           & $\bullet$        & $\bullet$        & $\bullet$        &                  & $\bullet$        & $\bullet$        & $\bullet$        & $\bullet$         \\
\hline
$\{=_p, =_i,<,\ip 3,\ip 4\}$             &                   &                    &                  &                  &                  &                     &                  &                  &                  &                  & $\bullet$        &                  &                  &               \\
\hline

\end{tabular}                                                                                                                                                                                                                                         \caption{$\mathsf{MIS}(r)$, for each $r\in\allr^+$; upper part: sets for which we give an explicit construction; lower part: symmetric ones. - Class: $\Unb$.}\label{tab:MISRunb}
\end{center}
\end{table}

\begin{lemma}
Tab.~\ref{tab:MISRunb} is correct.
\end{lemma}

\begin{proof}
Let $S$ be \underline{$\alli^{+}$}: proving that it is $=_p$-incomplete is very easy. Indeed, it suffices to take $\mathbb D=\mathbb D'=\mathbb Q$, $\zeta=(\zeta_p,\zeta_i)$, where $\zeta_i=Id_i$ (the identical relation on intervals), $\zeta_p=\{(0,1')\}$ plus the identical relation on points to have a surjective truth-preserving relation that breaks $=_p$. The $=_i$-incompleteness of \underline{$\{=_p,<,\ip 0,\ip 1\}$} is justified with the same argument used in the dense case (which was based on $\mathbb Q$). We can then prove that also \underline{$\{=_p,<,\ip 2,\ii 24\}$} is $=_i$-incomplete, by taking $\mathbb D=\mathbb D'=\mathbb Z$, $\zeta_p=Id_p$, and $\zeta_i=Id_i$ plus $\zeta_i([1,2],[0',1'])$. For the $<,\ip 0$-incompleteness of $\{=_i,=_p,\ip 1,\ii14\}$ we can recycle the argument used for the dense case (again, based on $\mathbb Q$). The $<$-incompleteness of \underline{$\{=_p\}\cup\alli^{+}$} can be proved by taking $\mathbb D=\mathbb D'=\mathbb Q$, $\zeta=(\zeta_p,\zeta_i)$, where $\zeta_i=Id_i$ and $\zeta_p(a)=-a'$ for every $a\in\mathbb Q$, which clearly respects all interval-interval relations and the equality between points. As for proving that \underline{$\{=_p,<\}\cup\alli^{+}$} is $m$-incomplete for each $m\in\allm^+$ we can recycle the same argument as in the dense case, as it was based on the set $\mathbb Q$. When $S$ is \underline{$\{=_p, =_i, \ip 2, \ii 04\}$}, we have to prove that it is $<,\ip 0,\ip1,i$-incomplete, where $i\in\alli^{+}\setminus\{\ii 04,=_i\}$. Consider two structures based on $\mathbb Q$, and let $\zeta=(\zeta_p,\zeta_i)$ be defined as $\zeta_p(a)=-a'$ for every point and $\zeta_i([a,b])=[-b',-a']$ for every interval. Clearly, {\em containment} is respected for both sorts; nevertheless, $<,\ip 0$ and $\ip 1$ and all interval-interval relations, except $\ii 04$,  are broken. Once again, we have already proved that \underline{$\{=_i, =_p, <, \ip3, \ip4, \ii03 \}$} is $\ip 0,\ip 1,\ip 2,i$-incomplete, where $i\in\alli^+\setminus\{\ii 03,=_i\}$ when we were treating the dense case, and the same holds for the $m,i$-incompleteness of \underline{$\{=_p, < \}\cup\alli^{+}$}, where $i\in\alli^{+}$ and $m\in\allm^{+}$. Consider now the $\ip 0,\ip 1,i$-incompleteness
of \underline{$\{=_p, =_i, <, \ip 2, \ii 24\}$}, where $i\in\alli^{+}\setminus\{\ii 24, =_i\}$. Take $\mathbb D=\mathbb D'=\mathbb Z$, $\zeta_p=Id_p$ and $\zeta_i([a,a+1])=[a'+1,a'+2]$ plus the identical relation over every other interval; since the only intervals affected by $\zeta$ are unitary, the relation $\ii 24$ cannot be broken, and since such interval do not have internal points, the relation $\ip 2$ cannot be broken either. Once more, the $\ii 14$-incompleteness of \underline{$\{=_i,=_p,<,\ip0,\ip 1\}$} comes directly from the dense case, and the same holds for the $r$-incompleteness
of \underline{$\{=_i, =_p, < ,\ii04\}$} for $r\in\alli^+\setminus\{\ii 04\}$, which concludes the proof.
\end{proof}

\section{Harvest: The Complete Picture for $\Den$ and $\Unb$}\label{sec:harvest}

We are now capable to identify all expressively different subsets of $\allr^{+}$ under the hypotheses of linearity+density and linearity+unboundedness. Unlike Part I, we limit ourselves to list the maximally incomplete sets and the minimally complete sets for each of the two cases in the full language only.

\begin{table}[t]
\small
\begin{center}
\begin{tabular}{|p{0.16\textwidth}|p{0.42\textwidth}|}
\hline
\multicolumn{2}{|c|}{$\allr^+$}\\
\hline
$\mathsf{mcs}$ & $\mathsf{MIS}$ \\
\hline
 $\{\ip0, \ip2 \}$     & $\{\ip0, \ip1, \ii14, =_i, =_p, < \}$                                 \\
  $\{\ip0, \ip3 \}$     &  $\{\ip2, \ii14, \ii24, \ii03, \ii34, \ii04, \ii44, =_i, =_p \}$   \\
  $\{\ip0, \ip4 \}$     &  $\{\ip3, \ip4, \ii03, =_i, =_p, < \}$     \\
  $\{\ip0, \ii24 \}$     &  $\{\ii14, \ii24, \ii03, \ii34, \ii04, \ii44, =_i, =_p, < \}$   \\
  $\{\ip0, \ii03 \}$     &    \\
  $\{\ip0, \ii34 \}$     &    \\
  $\{\ip0, \ii04 \}$     &                                                                   \\
  $\{\ip0, \ii44 \}$     &                                                                   \\
  $\{\ip1, \ip2 \}$     &                                                                   \\
  $\{\ip1, \ip3 \}$     &                                                                   \\
  $\{\ip1, \ip4 \}$     &                                                                   \\
  $\{\ip1, \ii24 \}$     &                                                                   \\
  $\{\ip1, \ii03 \}$     &                                                                   \\
  $\{\ip1, \ii34 \}$     &                                                                   \\
  $\{\ip1, \ii04 \}$     &                                                              \\
  $\{\ip1, \ii44 \}$     &                                                              \\
  $\{\ip2, \ip3 \}$     &                                                               \\
  $\{\ip2, \ip4 \}$     &                                                   \\
  $\{\ip2, < \}$     &                                                  \\
  $\{\ip3, \ii14 \}$     &                                                  \\
  $\{\ip3, \ii24 \}$     &                                                  \\
  $\{\ip3, \ii34 \}$     &                                                  \\
  $\{\ip3, \ii04 \}$     &                                                  \\
  $\{\ip3, \ii44 \}$     &                                                  \\
  $\{\ip4, \ii14 \}$     &                                                  \\
  $\{\ip4, \ii24 \}$     &                                                  \\
  $\{\ip4, \ii34 \}$     &                                                  \\
  $\{\ip4, \ii04 \}$     &                                                  \\
  $\{\ip4, \ii44 \}$   &                                                  \\
\hline
\end{tabular}
\caption{Minimally $\allr^+$-complete and maximally $\allr^+$-incomplete sets. - Class: $\Den$.}\label{tab:harvestden}
\end{center}
\end{table}

\begin{table}[t]
\small
\begin{center}
\begin{tabular}{|p{0.16\textwidth}|p{0.42\textwidth}|}
\hline
\multicolumn{2}{|c|}{$\allr^+$}\\
\hline
$\mathsf{mcs}$ & $\mathsf{MIS}$ \\
\hline
$\{\ip0, \ip2 \}$      & $\{\ip0, \ip1, \ii14, =_i, =_p, < \}$     \\
 $\{\ip0, \ip3 \}$     &  $\{\ip2, \ii24, =_i, =_p, < \}$     \\
 $\{\ip0, \ip4 \}$     &  $\{\ip2, \ii04, =_i, =_p \}$     \\
 $\{\ip0, \ii24 \}$    &  $\{\ip3, \ip4, \ii03, =_i, =_p, < \}$     \\
 $\{\ip0, \ii03 \}$    &  $\{\ii14, \ii24, \ii03, \ii34, \ii04, \ii44, =_i, =_p, < \}$   \\
 $\{\ip0, \ii34 \}$    &    \\
 $\{\ip0, \ii04 \}$    &                                                                   \\
 $\{\ip0, \ii44 \}$    &                                                                   \\
 $\{\ip1, \ip2 \}$     &                                                                   \\
 $\{\ip1, \ip3 \}$     &                                                                   \\
 $\{\ip1, \ip4 \}$    &                                                                   \\
 $\{\ip1, \ii24 \}$    &                                                                   \\
 $\{\ip1, \ii03 \}$    &                                                                   \\
 $\{\ip1, \ii34 \}$    &                                                                   \\
 $\{\ip1, \ii04 \}$    &                                                              \\
 $\{\ip1, \ii44 \}$    &                                                              \\
 $\{\ip2, \ip3 \}$    &                                                               \\
 $\{\ip2, \ip4 \}$    &                                                   \\
 $\{\ip2, \ii14 \}$    &                                                  \\
 $\{\ip2, \ii24, \ii04 \}$    &                                                  \\
 $\{\ip2, \ii03 \}$    &                                                  \\
 $\{\ip2, \ii34 \}$    &                                                  \\
 $\{\ip2, \ii04, < \}$    &                                                  \\
 $\{\ip2, \ii44 \}$    &                                                  \\
 $\{\ip3, \ii14 \}$    &                                                  \\
 $\{\ip3, \ii24 \}$    &                                                  \\
 $\{\ip3, \ii34 \}$    &                                                  \\
 $\{\ip3, \ii04 \}$    &                                                  \\
 $\{\ip3, \ii44 \}$    &                                                  \\
 $\{\ip4, \ii14 \}$    &                                                  \\
 $\{\ip4, \ii24 \}$    &                                                  \\
 $\{\ip4, \ii34 \}$    &                                                  \\
 $\{\ip4, \ii04 \}$    &                                                  \\
 $\{\ip4, \ii44 \}$    &                                                  \\
\hline
\end{tabular}
\caption{Minimally $\allr^+$-complete and maximally $\allr^+$-incomplete sets. - Class: $\Unb$.}\label{tab:harvestunb}
\end{center}
\end{table}

\begin{theorem}
If a set of relations is listed:
\begin{enumerate}[noitemsep,topsep=0pt]
\item as $\mathsf{mcs}(\allr^{+})$ in Tab.~\ref{tab:harvestden}, left column (resp., right column), then it is minimally $\allr^{+}$-complete (resp., maximally $\allr^+$-incomplete) in the class of all dense linearly ordered sets.
\item as $\mathsf{mcs}(\allr^{+})$ in Tab.~\ref{tab:harvestunb}, left column (resp., right column), then it is minimally $\allr^{+}$-complete (resp., maximally $\allr^+$-incomplete) in the class of all unbounded linearly ordered sets.
\end{enumerate}
\end{theorem}

\section{Conclusions}\label{sec:concl}

We considered here the two-sorted first-order temporal languages that includes relations between intervals, points, and inter-sort, and we treated equality between points and between intervals as any other relation, with no special role. Under four different assumptions on the underlying structure, namely, linearity only, linearity+discreteness, linearity+density, and linearity+unboundedness, we asked the question: which relation can be first-order defined by which subset of all relations? As a result, we identified all possible inter-definability between relations, all minimally complete, and all maximally incomplete subsets of relations. These inter-definability results allow one to effectively compute all expressively different subsets of relations, and, with minimal effort, also all expressively different subsets of relations for the interesting
sub-languages of interval relations only or mixed relations only. Two out of four interesting classes of linearly ordered sets are treated in Part I of this paper, while the remaining two are dealt with in the present one (Part II). There are several aspects of temporal reasoning in computer science to which this extensive study can be related:
\begin{enumerate}[noitemsep,topsep=0pt]
\item first-order logic over linear orders extended with temporal relations
  between points, intervals and mixed, is the very foundation of modal logics
  for temporal reasoning, and it is necessary to have a complete understanding
  of the former in order to deal with the latter.  Indeed these first-order
  languages and the second-order languages based on them are the
  \emph{correspondence languages} (over models and frames, respectively), for
  interval-based temporal logics. These logics have as yet a very imperfectly
  understood correspondence theory (see e.g.\ \cite{van2001correspondence} and \cite{UnifiedCor}), and a proper understanding of the families of correspondence languages and their relative expressivity as is developed in this paper, can be seen as an important step towards the development of such a theory. As a simple example, consider the following: because of the greater expressivity of first-order languages, there is more inter-definability between sets of temporal relations in first-order languages than between the corresponding sets of modalities in propositional temporal languages. This allows one to establish connections between propositional temporal languages which would not have been directly available via the study of the languages themselves;

\item automated reasoning techniques for interval-based modal logics are at their first stages; an uncommon, but promising approach is to treat them as pure
modal logics over particular Kripke-frames, whose first-order properties are, in fact, representation theorems such as those (indirectly) treated in this paper. As a future work, we also plan to systematically study the area of representation theorems;
\item the decidability of pure first-order theories extended with interval relations is well-known~\cite{lad87}; nevertheless, these results hinge on the decidability of MFO[$<$], while we believe that they could be refined both algorithmically and computationally;
\item the study of other related languages, important in artificial intelligence, can benefit from our results, such as first-order and modal logics for
spatial reasoning where basic objects are, for example, rectangles.
\end{enumerate}

\bibliographystyle{alpha}
\bibliography{biblio}

\end{document}